\documentclass[runningheads]{llncs}
\usepackage[T1]{fontenc}
\usepackage{graphicx}
\graphicspath{{images/}}
\usepackage{amssymb}
\usepackage{physics}
\begin{document}
\title{Efficient Algorithms for Quantum Hashing}
\author{Ilnar Zinnatullin\inst{1,2}\orcidID{0009-0005-6367-0049} \and
Kamil Khadiev\inst{1,2}\orcidID{0000-0002-5151-9908}}
\authorrunning{I. Zinnatullin and K.Khadiev}
\institute{Institute of Computational Mathematics and Information Technologies, Kazan Federal University, Kazan, Tatarstan, Russia \and
Zavoisky Physical-Technical Institute,
FRC Kazan Scientific Center of RAS, Kazan, Tatarstan, Russia\\
\email{IlnGZinnatullin@kpfu.ru}}
\maketitle              % typeset the header of the contribution
\begin{abstract}
Quantum hashing is a useful technique that allows us to construct memory-efficient algorithms and secure quantum protocols. First, we present a circuit that implements the phase form of quantum hashing using \(2^{n-1}\) \(CNOT\) gates, where \(n\) is the number of control qubits. Our method outperforms existing approaches and reduces the circuit depth. Second, we propose an algorithm that provides a trade-off between the number of \(CNOT\) gates (and consequently, the circuit depth) and the precision of rotation angles. This is particularly important in the context of NISQ (Noisy Intermediate-Scale Quantum) devices, where hardware-imposed angle precision limit remains a critical constraint.

\keywords{Quantum hashing  \and Quantum circuit decomposition \and Uniformly controlled rotation.}
\end{abstract}
\section{Introduction}

Nowadays, quantum hardware still belongs to the so-called Noisy Intermediate-Scale Quantum (NISQ) \cite{preskill2018} era, a term introduced by Preskill in 2018. This term reflects the fact that quantum computations are noisy and that the number of available qubits is limited. One of the major challenges is decoherence which disrupts the fragile quantum states necessary for computation and constrain the capabilities of the current quantum computers. Thus, the algorithm's execution time must be short enough before quantum states are broken.

Quantum computers provide a set of elementary gates. Typically, such a set includes one-qubit and two-qubit gates, and it is universal in the sense that any unitary transformation can be decomposed into a circuit consisting only of elementary gates \cite{barenco1995,divincenzo1995}. In most cases, the \(CNOT\) gate is used as the two-qubit gate. It is known that quantum gates are prone to errors. In particular, two-qubit gates have higher error rates compared to their one-qubit counterparts \cite{urlibm2024}. So, minimizing the total number of \(CNOT\) gates is essential for improving reliability. Hardware-imposed angle precision limit \cite{koczor2024} is another problem that we need to take into consideration. Because of limited angle precision, some small-angle rotations cannot be performed accurately on the current hardware. As a result, we face numerous challenges when implementing quantum algorithms on NISQ devices. Therefore, it is crucial to design quantum algorithms optimized according to various metrics, such as the circuit depth, the number of \(CNOT\) gates and the rotation angle precision.

Hashing is a well‑established technique used in a variety of computational and cryptographic scenarios. In classical computing, a hash function maps an input of arbitrary length to a fixed‑length output (a hash), enabling fast search, data integrity check, etc. In cryptography, hash functions must be one‑way and collision‑resistant, making them essential for digital signatures, message authentication codes (MACs), secure password storage, and blockchain security. A quantum hash function \cite{akvz2025} is a classical-quantum one-way function that maps classical inputs to quantum states in such a way that states corresponding to different inputs are `nearly' orthogonal. Because quantum hashes are exponentially smaller than the original inputs, quantum hash functions are non-invertible by virtue of fundamental quantum information theory. Moreover, collision resistance allows us to distinguish different quantum hashes with high probability.

Ambainis and Freivalds \cite{ambainis1998} proposed a technique called quantum fingerprinting to construct quantum automata recognizing the unary $MOD_{p}=\{a^i: i$ $mod$ $p=0\}$ language, where \(p\) is prime. Then, Buhrman et al. \cite{buhrman2001} used quantum fingerprinting that is based on binary error correcting codes to construct communication protocol for equality problem. Ablayev and Vasiliev \cite{ablayev2014} proposed a non-binary classical-quantum one-way function as a generalization of quantum fingerprinting and introduced the notion of quantum hashing. Later, Vasiliev proposed a phase form approach for quantum hashing in \cite{vasiliev2016}. Efficient construction of branching programs for quantum hashing in terms of both the circuit depth and width was investigated in \cite{ablayev2018}. In \cite{ablayev2018}, it is stated that when \(\varepsilon\)-biased sets are used to hash elements of \(\mathbb{Z}_{q}\), the algorithm has the running time of \(\Omega(\log q)\) and requires \(\Omega(\log\log q)\) qubits. 
The quantum hashing (fingerprinting) approach has been widely used in various areas such as 
stream processing algorithms \cite{l2009}, query model algorithms \cite{aaksv2022,asa2024}, online algorithms \cite{kk2019disj,kk2022}, branching programs \cite{kkk2022,aakv2018}, development of quantum devices \cite{v2016model}, automata \cite{ambainis1998,Ambainis2009}, etc.

In this paper, we focus on designing efficient quantum circuit (a program for quantum computer) for quantum hashing. By ``efficient'', we mean a circuit with the minimal number of \(CNOT\) gates. The base element of quantum circuit for the method is the uniformly controlled rotation gate, for which the most efficient representation was suggested by M\(\ddot{o}\)tt\(\ddot{o}\)nen et al. \cite{mottonen2004}. The technique allows us to represent the uniformly controlled rotation (and, as a result, a circuit for quantum hashing) with $n$ control qubits using $2^n$ \(CNOT\) gates.

We consider the phase form quantum hashing algorithm (the uniformly controlled rotation with a specific initial state), and present a circuit with \(2^{n-1}\) \(CNOT\) gates. Reducing the number of \(CNOT\) gates by half is important as two-qubit gates are more costly to implement on the current quantum devices. 

Another issue with representing uniformly controlled rotation by our technique and technique of M\(\ddot{o}\)tt\(\ddot{o}\)nen et al. \cite{mottonen2004} is rotation angle precision. These techniques require much higher precision for the modified angles compared to the original ones. We provide an algorithm that demonstrates a trade-off between the number of \(CNOT\) gates and the rotation angle precision. Similar trade-off was demonstrated by Khadieva et al. \cite{khadieva2024}. Asymptotically, our algorithm is equivalent to the existing one, yet it requires fewer \(CNOT\) gates in practice.

Optimization the circuit for quantum hashing algorithm is considered from different points of view. In \cite{zinnatullin2023}, the authors adapt quantum circuits implementing quantum hashing for specific quantum processor architectures and optimize them with respect to the number of \(CNOT\) gates. Researchers consider shallow circuits for approximate versions of quantum hashing in \cite{ziiatdinov2023gaps,zkk2025,v2023}.

The paper is organized as follows. Section~\ref{sec:preliminaries} provides the necessary preliminaries.  In Section~\ref{sec:result1}, we demonstrate how to construct an efficient quantum circuit that implements the phase form of quantum hashing. Section~\ref{sec:result2} is devoted to the algorithm that offers a trade-off between the number of \(CNOT\) gates and the precision of the rotation angles. Finally, we summarize our results in Section~\ref{sec:conclusions}.

\section{Preliminaries}\label{sec:preliminaries}

\textbf{Quantum computation.}
We use Dirac notation. In general case, a qubit is represented as a column vector (ket vector) \(\ket{\psi} = \alpha\ket{0} + \alpha_{1}\ket{1} \in \mathcal{H}^{2}\) with complex amplitudes satisfying \(|\alpha_{0}|^{2} + |\alpha_{1}|^{2} = 1\). A state of \(n\) qubits is described by a ket vector from \((\mathcal{H}^2)^{\otimes n}\) and has the form \(\sum_{j=0}^{2^{n}-1}\alpha_{j}\ket{j}\), where \(\alpha_{j} \in \mathbb{C}\) and \(\sum_{j=0}^{2^{n}-1}|\alpha_{j}|^{2} = 1\).

Quantum circuit model is a model for quantum computation, similar to classical circuits, which consists of qubits (represented as horizontal lines), quantum gates corresponding to unitary transformations and measurements to extract classical information from the qubits. A quantum circuit is a quantum algorithm that is characterized by two parameters: depth and width. The circuit depth describes the time complexity of the algorithm, while the width corresponds to its space complexity (i.e. the number of qubits required). More details about quantum circuits can be found in \cite{nielsen2010}.

Next, we present matrix representations of the gates used in our work: negation gate \(X =
\begin{pmatrix}
	0 & 1\\
	1 & 0 
\end{pmatrix}\), Hadamard gate \(H =
\begin{pmatrix}
\frac{1}{\sqrt{2}} & \frac{1}{\sqrt{2}}\\
\frac{1}{\sqrt{2}} & -\frac{1}{\sqrt{2}} 
\end{pmatrix}\), rotation about the \(z\)-axis \(R_{z}(\theta) =
\begin{pmatrix}
	e^{-\frac{i\theta}{2}} & 0\\
	0 & e^{\frac{i\theta}{2}} 
\end{pmatrix}\), rotation about the \(y\)-axis \(R_{y}(\theta) =
\begin{pmatrix}
\cos(\frac{\theta}{2}) & -\sin(\frac{\theta}{2})\\
\sin(\frac{\theta}{2}) & \cos(\frac{\theta}{2}) 
\end{pmatrix}\), relative phase shift \(P(\theta) = 
\begin{pmatrix}
	1 & 0             \\
	0 & e^{\frac{i\theta}{2}} 
\end{pmatrix}
\), two-qubit gate \(CNOT = \begin{pmatrix}
	1 & 0 & 0 & 0\\
	0 & 1 & 0 & 0\\
	0 & 0 & 0 & 1\\
	0 & 0 & 1 & 0 
\end{pmatrix}.\)

By \(C^{n-1}(R)\), we denote an $n$-qubit controlled \(R\) gate with \(n-1\) control qubits that applies \(R\) to the qubit with index \(n\) if and only if all control qubits in the state \(\ket{1}\), where \(R = 
\begin{pmatrix}
	r_{11} & r_{12}\\
	r_{21} & r_{22} 
\end{pmatrix}.
\) 	The \(2^{n} \times 2^{n}\) matrices corresponding to the \(C^{n-1}(R_{z}(\theta))\) and \(C^{n-1}(R_{y}(\theta))\) gates are
\begin{equation}\label{eq:controlled_rotation_matrix}
% C^{n-1}(R_{z}(\theta)) =
% \begin{pmatrix}
% 	1 &        &        &      \\
% 	  & \ddots &        &      \\
% 	  &        & 1  &      \\        
% 	&        &        & R_{z}(\theta)
% \end{pmatrix},\quad
% C^{n-1}(R_{y}(\theta)) =
% \begin{pmatrix}
% 	1 &        &    &\\
% 	  & \ddots &    &\\
% 	  &        & 1  &\\        
% 	  &        &    & R_{y}(\theta)
% \end{pmatrix}.
\scalebox{1.0}{$
C^{n-1}(R_{z}(\theta)) =
\begin{pmatrix}
	1 &        &        &      \\
	  & \ddots &        &      \\
	  &        & 1  &      \\        
	&        &        & R_{z}(\theta)
\end{pmatrix},\quad
C^{n-1}(R_{y}(\theta)) =
\begin{pmatrix}
	1 &        &    &\\
	  & \ddots &    &\\
	  &        & 1  &\\        
	  &        &    & R_{y}(\theta)
\end{pmatrix}. 
$}
\end{equation}
 
By \(UCR^{n-1}_{a}\), we denote an \(n\)-qubit uniformly controlled rotation about the \(a\)-axis \cite{mottonen2004} that employs \(n-1\) control qubits and uses all possible control states to rotate the qubit with index \(n\).

\textbf{Quantum hashing.}
Let \(S = \{s_{0}, s_{1}, \ldots, s_{d-1}\} \subseteq \mathbb{Z}_{q}\) be an \(\varepsilon\)-biased set, i.e., set of parameters satisfying \(
\frac{1}{d}\left|\sum_{j=0}^{d-1} e^{i2\pi s_{j}x/q}\right| \leq \varepsilon\), for every \(x \in \mathbb{Z}_{q} \backslash \{0\}\).

For \(x \in \mathbb{Z}_{q}\), we define its quantum hash in amplitude and phase forms. An \((n-1)\)-qubit quantum hash in the phase form \cite{vasiliev2016} is given by

\begin{equation}\label{eq:qh_phase_form}
	% \ket{\psi(x)} = \frac{1}{\sqrt{d}}\sum_{j=0}^{d-1} e^{i2\pi s_{j}x/q}\ket{j}.
\scalebox{1.0}{$
\ket{\psi(x)} = \frac{1}{\sqrt{d}}\sum_{j=0}^{d-1} e^{i2\pi s_{j}x/q}\ket{j}.
$}
\end{equation}

An \(n\)-qubit quantum hash in the amplitude form \cite{ablayev2014} is defined as

\begin{equation}\label{eq:qh_amplitude_form}
	% \ket{\psi(x)} = \frac{1}{\sqrt{d}}\sum_{j=0}^{d-1}\ket{j}\left(\cos\left(\frac{2\pi s_{j}x}{q}\right)\ket{0}+\sin\left(\frac{2\pi s_{j}x}{q}\right)\ket{1}\right).
\scalebox{1.0}{$
\ket{\psi(x)} = \frac{1}{\sqrt{d}}\sum_{j=0}^{d-1}\ket{j}\left(\cos\left(\frac{2\pi s_{j}x}{q}\right)\ket{0}+\sin\left(\frac{2\pi s_{j}x}{q}\right)\ket{1}\right).
$}
\end{equation}

Note that \(n - 1 = \log d\) and \(d = O\left(\frac{\log q}{\varepsilon^{2}}\right)\) \cite{vasiliev2016,ablayev2014}. We give the following formula for quantum hash that is a generalization of Equations~(\ref{eq:qh_phase_form}) and (\ref{eq:qh_amplitude_form}):

\begin{equation}\label{eq:qh_generalization}
	% \ket{\psi(x)} = \frac{1}{\sqrt{d}}\sum_{j=0}^{d-1}\ket{j}\left(R_{a}\left(\theta_{j}\right)\ket{q_{n}}\right),
\scalebox{1.0}{$
\ket{\psi(x)} = \frac{1}{\sqrt{d}}\sum_{j=0}^{d-1}\ket{j}\left(R_{a}\left(\theta_{j}\right)\ket{q_{n}}\right),
$}
\end{equation}
where \(R_{a}\) is a rotation about the \(a\)-axis on the Bloch sphere, \(\theta_{j} = \frac{4\pi s_{j}x}{q}\). For the phase form, we use \(R_{z}\) gates and the target qubit \(\ket{q_{n}} = \ket{1}\). For the amplitude form,  we use \(R_{y}\) gates and the target qubit \(\ket{q_{n}} = \ket{0}\). Note that in Equation~(\ref{eq:qh_generalization}) an ancilla (namely, the target qubit \(\ket{q_{n}}\)) is used to create a quantum hash while Equation~(\ref{eq:qh_phase_form}) lacks it.
An algorithm for quantum hashing with parameter \( \tilde{\theta} = (\theta_{0}, \theta_{1}, \ldots, \theta_{d-1})\) in the quantum circuit model is presented in Figure~\ref{fig:qh_algorithm}. For further details on quantum hashing, please refer to \cite{bookablayev2015,bookablayev2023}.

\begin{figure}[h!]
	\centering
	\includegraphics[scale=0.3]{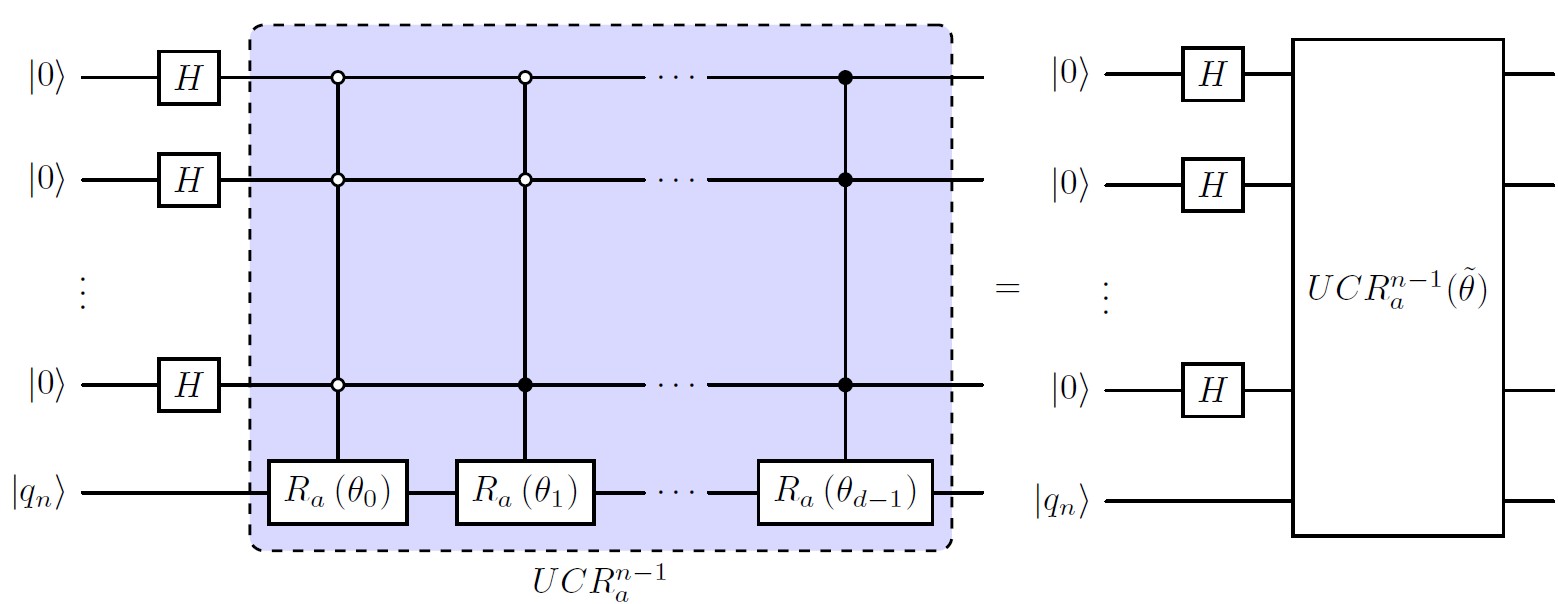}%[scale=0.38]
	\caption{Algorithm for quantum hashing.}
	\label{fig:qh_algorithm}
\end{figure}

\subsection{Efficient decomposition of \(UCR\) gates}\label{sec:efficient_decomposition_of_ucr}

M\(\ddot{o}\)tt\(\ddot{o}\)nen et al. presented a method \cite{mottonen2004} for decomposing a \(UCR\) gate into a circuit consisting of an alternating sequence of single-qubit rotations and \(CNOT\) gates. This method requires a modification of the original rotation angles.
The decomposition is constructed recursively by applying the decomposition step from \cite{mottonen2006} (see Figure~\ref{fig:ucr_decomposition_step}) to each \(UCR\) gate until we get a circuit consisting only of one-qubit rotations and \(CNOT\) gates. The parameters of the \(UCR\) gates shown in Figure~\ref{fig:ucr_decomposition_step} are defined as follows: \(\tilde{\theta}^{0}_{1} = (\theta^{0}_{0}=\theta_{0}, \theta^{0}_{1}=\theta_{1}, \ldots, \theta^{0}_{2r-1}=\theta_{2r-1})\), \( \tilde{\theta}^{1}_{1} = (\theta^{1}_{0}, \theta^{1}_{1}, \ldots, \theta^{1}_{r-1}), \tilde{\theta}^{1}_{2} = (\theta^{2}_{0}, \theta^{2}_{1}, \ldots, \theta^{2}_{r-1})\), where \(r = 2^{n-2}\).

\begin{figure}[h!]
	\centering
	\includegraphics[scale=0.25]{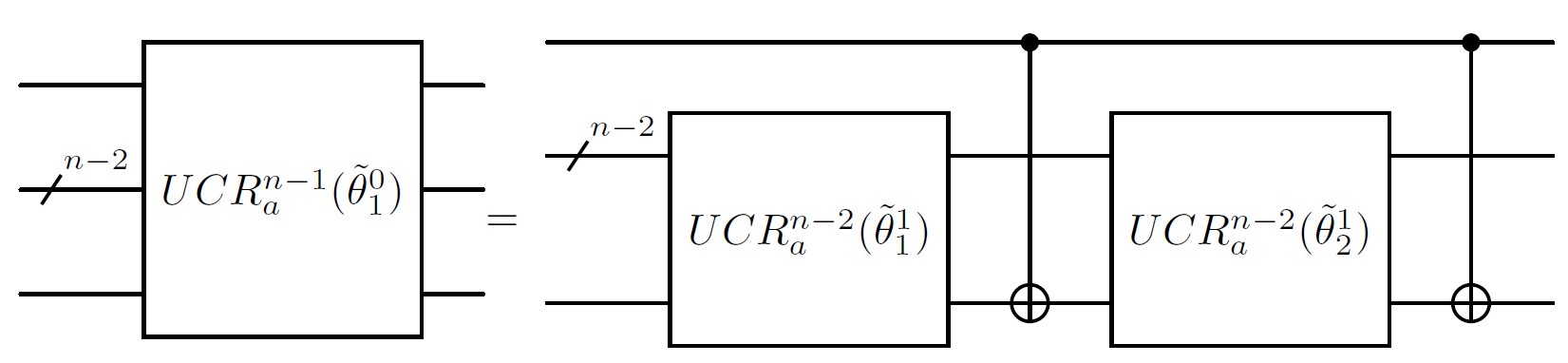}%[scale=0.31]
	\caption{Decomposition step for M\(\ddot{o}\)tt\(\ddot{o}\)nen et al.'s method.}
	\label{fig:ucr_decomposition_step}
\end{figure}

The connection between the original and the modified angles for the decomposition step is the following: \(\theta^{0}_{0} = \theta^{1}_{0} + \theta^{2}_{0}, \ldots, \theta^{0}_{r-1} = \theta^{1}_{r-1} + \theta^{2}_{r-1}, \theta^{0}_{r} = \theta^{1}_{0} - \theta^{2}_{0}, \ldots, \theta^{0}_{2r-1} = \theta^{1}_{r-1} - \theta^{2}_{r-1}\).
%\begin{equation*}
%\scalebox{0.95}{$
%	\begin{split}
%		\theta^{0}_{0} &= \theta^{1}_{0} + \theta^{2}_{0},\\
%		&\ldots\\
%		\theta^{0}_{r-1} &= \theta^{1}_{r-1} + \theta^{2}_{r-1},\\
%		\theta^{0}_{r} &= \theta^{1}_{0} - \theta^{2}_{0},\\
%		&\ldots\\
%		\theta^{0}_{2r-1} &= \theta^{1}_{r-1} - \theta^{2}_{r-1}.\\
%	\end{split}   
%$} 
%\end{equation*}
So, we get
$
\theta^{1}_{i} = \frac{\theta^{0}_{i} + \theta^{0}_{r+i}}{2} = \frac{\theta_{i} + \theta_{r+i}}{2},\ \theta^{2}_{i} = \frac{\theta^{0}_{i} - \theta^{0}_{r+i}}{2} = \frac{\theta_{i} - \theta_{r+i}}{2},
$
where $i \in [0, r-1]$.

Finally, we obtain a circuit containing \(2^{n-1} = d\) one-qubit \(R_{a}\) gates and \(2^{n-1} = d\) \(CNOT\) gates. Let \(\tilde{\theta'} = (\theta'_{0}, \theta'_{1}, \ldots, \theta'_{d-1})\) denote the vector of modified angles in the resulting circuit. The relationship between the angles is given by
\begin{equation}\label{eq:connection_between_angles}
\scalebox{1.0}{$
\tilde{\theta'} = \frac{1}{d}M^{T}\tilde{\theta},
$}
\end{equation}
where \(M\) is a \(d \times d\) matrix with entries \(M_{ij} =(-1)^{(b_{i-1} \cdot g_{j-1})}, 1 \leq i,j \leq d\), and \((b_{i-1} \cdot g_{j-1})\) denotes the scalar product of the \((i-1)\)th codeword of the standard binary code and the \((j-1)\)th codeword of Gray code. Figure~\ref{fig:ucr_decomposition_example} shows an example for \(d = 8\).

\begin{figure}[h!]
	\centering
	\includegraphics[scale=0.36]{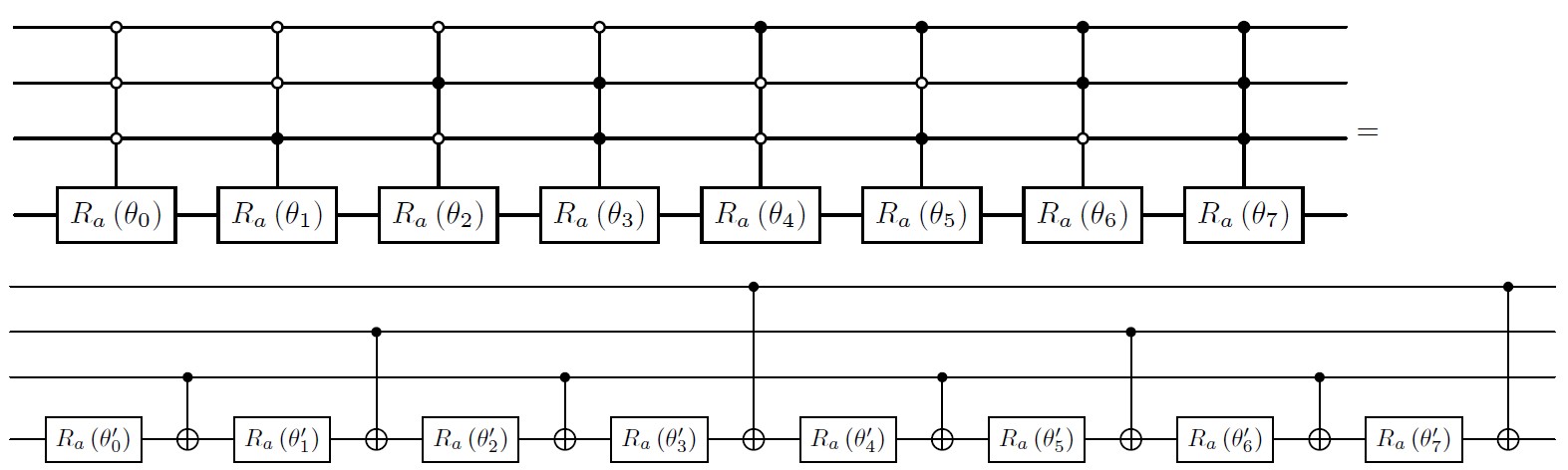}%[scale=0.46]
	\caption{Decomposition for \(d = 8\).}
	\label{fig:ucr_decomposition_example}
\end{figure}

\section{Circuit optimization}\label{sec:result1}

In this section, we propose an algorithm that reduces the number of \(CNOT\) gates in the circuit that implements the phase form of quantum hashing by half compared to the best existing method.
%
%\subsection{Preliminaries for circuit optimization}

Firstly, we present lemmas that show several equivalences of quantum circuits. These results will be used in the circuit optimization process.

\begin{lemma}\label{lemma:lemma_1}
	The circuit equivalence presented in Figure~\ref{fig:lemma1} holds.
	\begin{figure}[h!]
		\centering
		\includegraphics[scale=0.29]{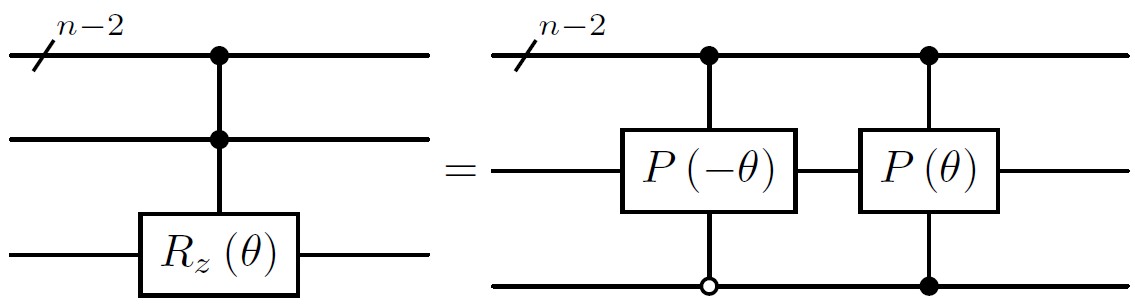}%[scale=0.34]
		\caption{Circuit equivalence for the \(n\)-qubit controlled \(R_{z}(\theta)\) rotation.}
		\label{fig:lemma1}
	\end{figure}
\end{lemma}

\begin{proof}
	Using Equation~\ref{eq:controlled_rotation_matrix}, we note that only two basis states undergo a nontrivial effect:
	\[
	\begin{split}
		&\ket{\underline{111}\ldots\underline{11}0} \mapsto e^{-i\theta/2}\ket{111\ldots110},\\
		&\ket{\underline{111}\ldots\underline{11}1} \mapsto e^{i\theta/2}\ket{111\ldots111}, 
	\end{split}
	\]
	with the control qubits underlined.

One can easily verify that
\[
\begin{pmatrix}
	1 &   &   &             \\
	  & 1 &   &             \\
	  &   & e^{-i\theta/2} &\\        
	  &   &                & e^{i\theta/2} 
\end{pmatrix}=
\begin{pmatrix}
	1 &   &   &             \\
	& 1 &   &             \\
	&   & 1 &\\        
	&   &                & e^{i\theta/2} 
\end{pmatrix}
\begin{pmatrix}
	1 &   &   &             \\
	& 1 &   &             \\
	&   & e^{-i\theta/2} &\\        
	&   &                & 1 
\end{pmatrix}.
\]
	
	So, we propose an alternative interpretation of \(C^{n-1}(R_{z}(\theta))\) operator, namely, consecutive applications of the following two operators:
	\[
	\begin{split}
		&\ket{\underline{111}\ldots\underline{1}0\underline{0}} \mapsto \ket{111\ldots100},\\
		&\ket{\underline{111}\ldots\underline{1}1\underline{0}} \mapsto e^{-i\theta/2}\ket{111\ldots110},
	\end{split}
	\]
	and
	\[
	\begin{split}
		&\ket{\underline{111}\ldots\underline{1}0\underline{1}} \mapsto \ket{111\ldots101},\\
		&\ket{\underline{111}\ldots\underline{1}1\underline{1}} \mapsto e^{i\theta/2}\ket{111\ldots111}.
	\end{split}
	\]As before, all other basis states remain unchanged, and the control qubits are underlined. In this interpretation, two \(n\)-qubit controlled relative phase shifts are applied to the qubit with index \(n-1\). The circuit shown in Figure~(\ref{fig:multi_qubit_relative_phase_shifts}) illustrates this.
	\begin{figure}[h!]
		\centering
		\includegraphics[scale=0.25]{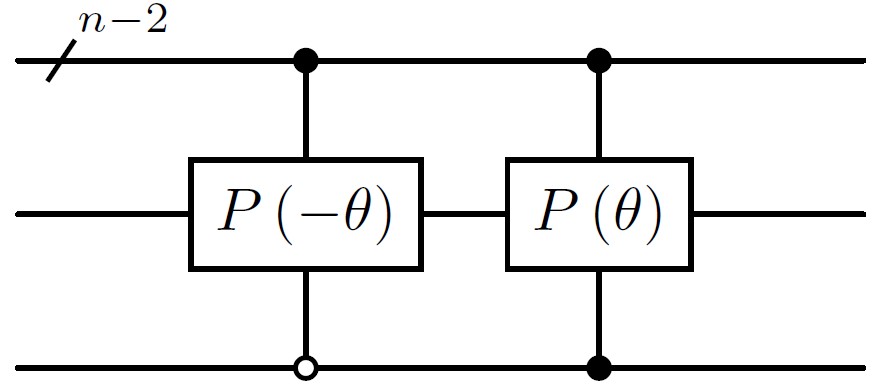}
		\caption{Consecutive \(n\)-qubit controlled relative phase shifts.}
		\label{fig:multi_qubit_relative_phase_shifts}
	\end{figure}
\end{proof}

\begin{lemma}\label{lemma:lemma_2}
	Relative phase shift $P$ is equivalent, up to a global phase factor, to the half-argument rotation $R_{z}$, that is, \(P(\theta) = e^{i\theta/4}R_{z}(\theta/2)\).
\end{lemma}

\begin{proof}
	Obviously,
	$
	P(\theta) = 
	\begin{pmatrix}
		1 & 0             \\
		0 & e^{i\theta/2} 
	\end{pmatrix}=
	e^{i\theta/4}
	\begin{pmatrix}
		e^{-i\theta/4} & 0             \\
		0              & e^{i\theta/4} 
	\end{pmatrix}=
	e^{i\theta/4}R_{z}(\theta/2).
	$
\end{proof}

\begin{lemma}\label{lemma:lemma_3}
	The circuit equivalence shown in Figure~\ref{fig:lemma3} holds. In mathematical form, this equivalence is expressed as
	$
	(I^{\otimes (n-1)}\otimes X)C^{n-1}(R_{z}(\theta)) = C^{n-1}(R_{z}(-\theta))(I^{\otimes (n-1)}\otimes X).
	$
	\begin{figure}[h!]
		\centering
		\includegraphics[scale=0.2]{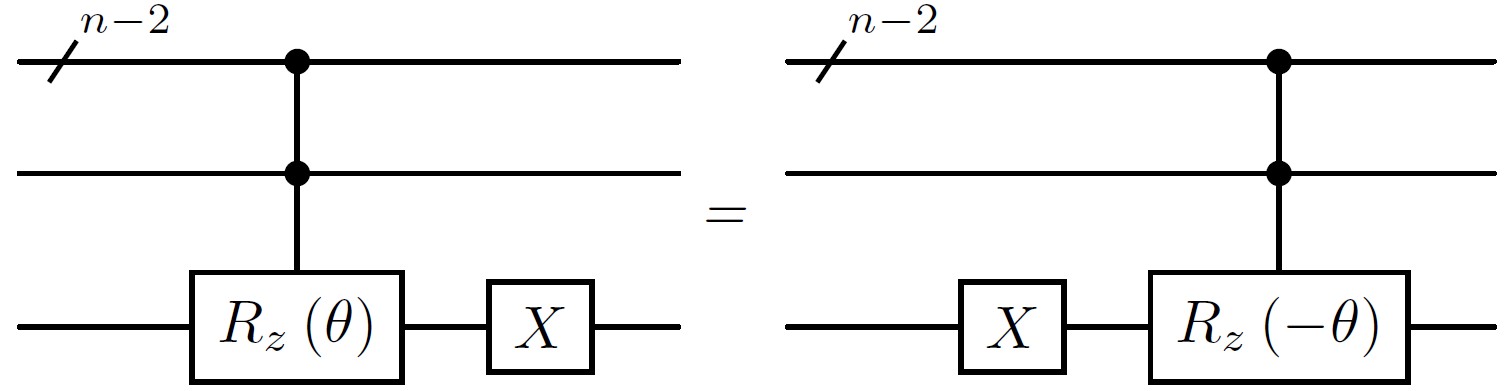}%[scale=0.3]
		\caption{Circuit equivalence for \(n\)-qubit controlled \(R_{z}(\theta)\) rotation and single \(X\) gate.}
		\label{fig:lemma3}
	\end{figure}
\end{lemma}

\begin{proof}
	Let us consider the matrix representations corresponding to the left and the right circuits.
	
	% \[
	% I \otimes X =
	% \begin{pmatrix}
		% X &\\
		%   & X
		% \end{pmatrix}_{4 \times 4}
	% \]
	
	The matrix of the $n$-qubit operator that applies $X$ to the qubit with index $n$ is given by
	\[
	\underbrace{I \otimes I \otimes \ldots \otimes I \otimes I}_{n-1} \otimes X =
	\underbrace{I \otimes I \otimes \ldots \otimes I}_{n-2} \otimes
	\begin{pmatrix}
		X &   \\
		& X 
	\end{pmatrix} = 
	\begin{pmatrix}
		X &        &  \\
		& \ddots &  \\
		&        & X 
	\end{pmatrix}_{2^{n} \times 2^{n}}.
	\]
	
	The matrix associated with the left circuit is
	\[
	\scalebox{0.9}{$
\begin{pmatrix}
	X &        &   &   \\
	& \ddots &   &   \\
	&        & X &   \\
	&        &   & X 
\end{pmatrix}
\begin{pmatrix}
	I &        &   &              \\
	& \ddots &   &              \\
	&        & I &              \\        
	&        &   & R_{z}(\theta) 
\end{pmatrix}=
\begin{pmatrix}
	X &        &   &                      \\
	& \ddots &   &                      \\
	&        & X &                      \\
	&        &   & \tilde{R}_{z}(\theta) 
\end{pmatrix}.
	$}
	\]
	
	The matrix that represents the right circuit is
	\[
	\scalebox{0.9}{$
	\begin{pmatrix}
		I &        &   &                  \\
		& \ddots &   &               \\
		&        & I &               \\        
		&        &   & R_{z}(-\theta) 
	\end{pmatrix}
	\begin{pmatrix}
		X &        &   &   \\
		& \ddots &   &   \\
		&        & X &   \\
		&        &   & X 
	\end{pmatrix}=
	\begin{pmatrix}
		X &        &   &                      \\
		& \ddots &   &                      \\
		&        & X &                      \\
		&        &   & \tilde{R}_{z}(\theta) 
	\end{pmatrix}.
	$}
	\]
	
	We denote by $\tilde{R}_{z}(\theta)$ a matrix
	\[
	\begin{pmatrix}
		0              & e^{i\theta/2} \\
		e^{-i\theta/2} & 0             
	\end{pmatrix}.
	\]
	
	Direct comparison of the matrices reveals that they are identical.
\end{proof}

%\subsection{Circuit optimization algorithm}

Our circuit optimization algorithm that reduces the number of \(CNOT\) gates is based on the next theorem.

\begin{theorem}
The phase version of the quantum hashing algorithm can be represented as a quantum circuit with $2^{n-1}$ CNOT gates.
\end{theorem}

\begin{proof}
For proving the claim of the theorem, we should show the equivalence that presented in Figure~\ref{fig:final_result}. So, in that case, we can use technique from Section \ref{sec:efficient_decomposition_of_ucr} for representation of $UCR$ gate and obtain a circuit with $2^{n-1}$ \(CNOT\) gates.

\begin{figure}[h!]
	\centering
	\includegraphics[scale=0.46]{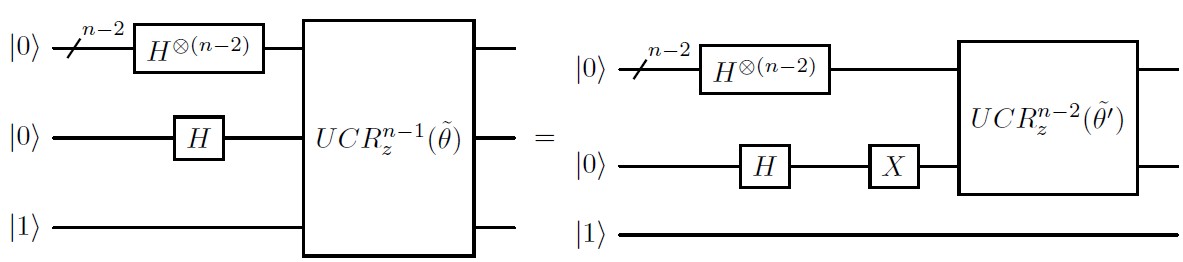}%[scale=0.57]
	\caption{Elimination of the ancilla in the circuit implementing the phase form of quantum hashing.}
	\label{fig:final_result}
\end{figure}

We start with the original circuit shown in Figure~\ref{fig:circuit_for_the_phase_form_of_qh}. The objective is to eliminate the ancilla qubit (i.e. the qubit with index \(n\)), thereby reducing the number of control qubits in the \(UCR^{n-1}_{z}\) gate by one.  Subsequently, we apply the method described in Section~\ref{sec:efficient_decomposition_of_ucr} to the resulting \(UCR^{n-2}_{z}\) gate and get a circuit consisting of \(2^{n-2}\) \(CNOT\) gates.

\begin{figure}[h!]
	\centering
	\includegraphics[scale=0.18]{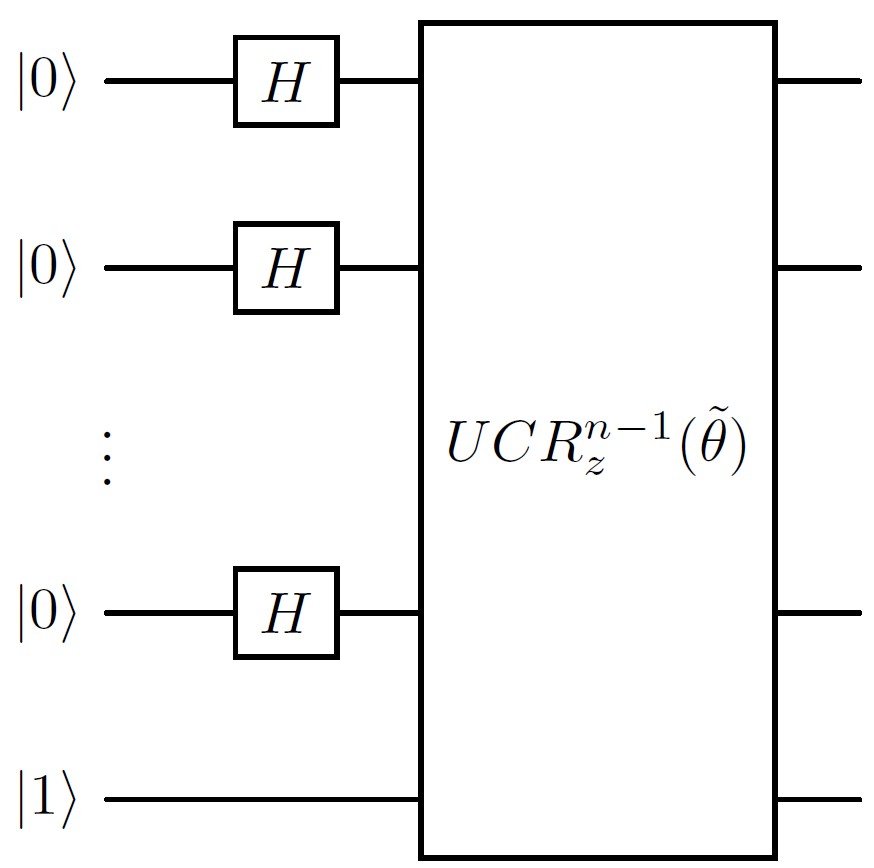}%[scale=0.26]
	\caption{Algorithm for the phase form of quantum hashing.}
	\label{fig:circuit_for_the_phase_form_of_qh}
\end{figure}

Let us reorder \(n\)-qubit controlled rotations in the \(UCR_{z}^{n-1}\) gate using Gray code instead of the standard binary code (see Figure~\ref{fig:decomposition_of_ucr_using_gray}). This reordered structure is then applied to the circuit presented in Figure~\ref{fig:circuit_for_the_phase_form_of_qh} (see Figure~\ref{fig:intermediate_decomposition}). We use Gray code because adjacent codewords differ in exactly one position. This allows us to switch between states of control qubits by applying a single \(X\) gate. Note that the circuit shown in Figure~\ref{fig:intermediate_decomposition} consists solely of \(n\)-qubit controlled rotations of the form \(C^{n-1}(R_{z})\), which require all control qubits to be in the state \(\ket{1}\).

For simplicity, we fix \(n = 4\) (see Figure~\ref{fig:circuit_for_n=4}) and follow the optimization process step by step. Our focus is on optimizing the red border-boxed gate depicted in Figure~\ref{fig:circuit_for_n=4}. First, we apply Lemma~\ref{lemma:lemma_1} to each blue border-boxed gate in Figure~\ref{fig:circuit_for_n=4} and get a circuit shown in Figure~\ref{fig:after_application_of_lemma_1}. Since the last qubit is in the state \(\ket{1}\), multi-qubit controlled relative phase shifts with negative arguments are not applied. Consequently, the circuit is further simplified to the one shown in Figure~\ref{fig:simplified_circuit}. Next, we replace relative phase shifts with rotations about the \(z\)-axis using Lemma~\ref{lemma:lemma_2} (see Figure~\ref{fig:simplified_circuit_with_rotations}). Subsequently, we apply Lemma~\ref{lemma:lemma_3} to the blue border-boxed segments of the circuit shown in Figure~\ref{fig:simplified_circuit_with_rotations}. This step reduces the number of multi-qubit controlled rotations by half and eliminates some of the \(X\) gates, as each blue border-boxed pair of multi-qubit controlled rotations in Figure~\ref{fig:merging_multi_qubit_controlled_rotations} can be merged into one multi-qubit controlled rotation. Consecutive \(X\) gates (that are in red border-boxes) applied to the third qubit cancel out each other. So, we obtain the circuit shown in Figure~\ref{fig:circuit_ucr_n-2_using_gray}. 
Careful analysis reveals that Gray code is used to iterate all multi-qubit controlled rotations in Figure~\ref{fig:circuit_ucr_n-2_using_gray}. Alternatively, we can use the standard binary code for this purpose, resulting in the circuit presented in Figure~\ref{fig:circuit_ucr_n-2}. As a result, we obtain the circuit identity shown in Figure~\ref{fig:ucr_identity_small}.

\begin{figure}[h!]
	\centering
	\includegraphics[scale=0.4]{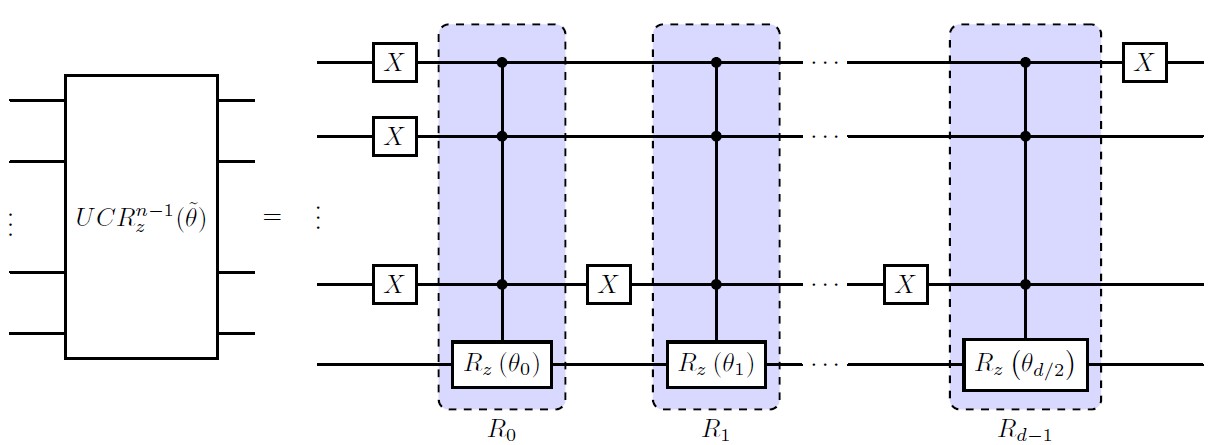}%[scale=0.53]
	\caption{Decomposition of \(UCR^{n-1}_{z}\) gate using Gray code.}
	\label{fig:decomposition_of_ucr_using_gray}
\end{figure}

\begin{figure}[h!]
	\centering
	\includegraphics[scale=0.27]{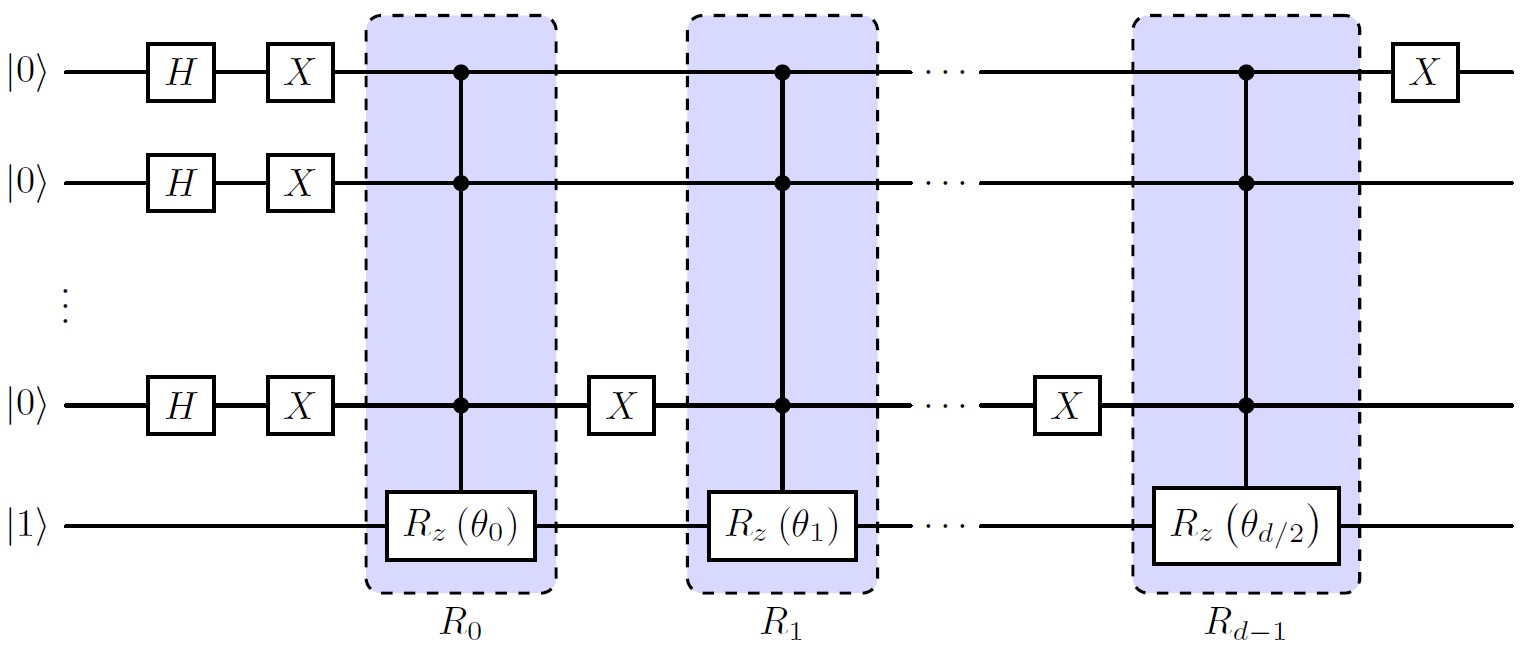}%[scale=0.38]
	\caption{Intermediate decomposition for the circuit implementing the phase form of quantum hashing.}
	\label{fig:intermediate_decomposition}
\end{figure}

\begin{figure}[h!]
	\centering
	\includegraphics[scale=0.38]{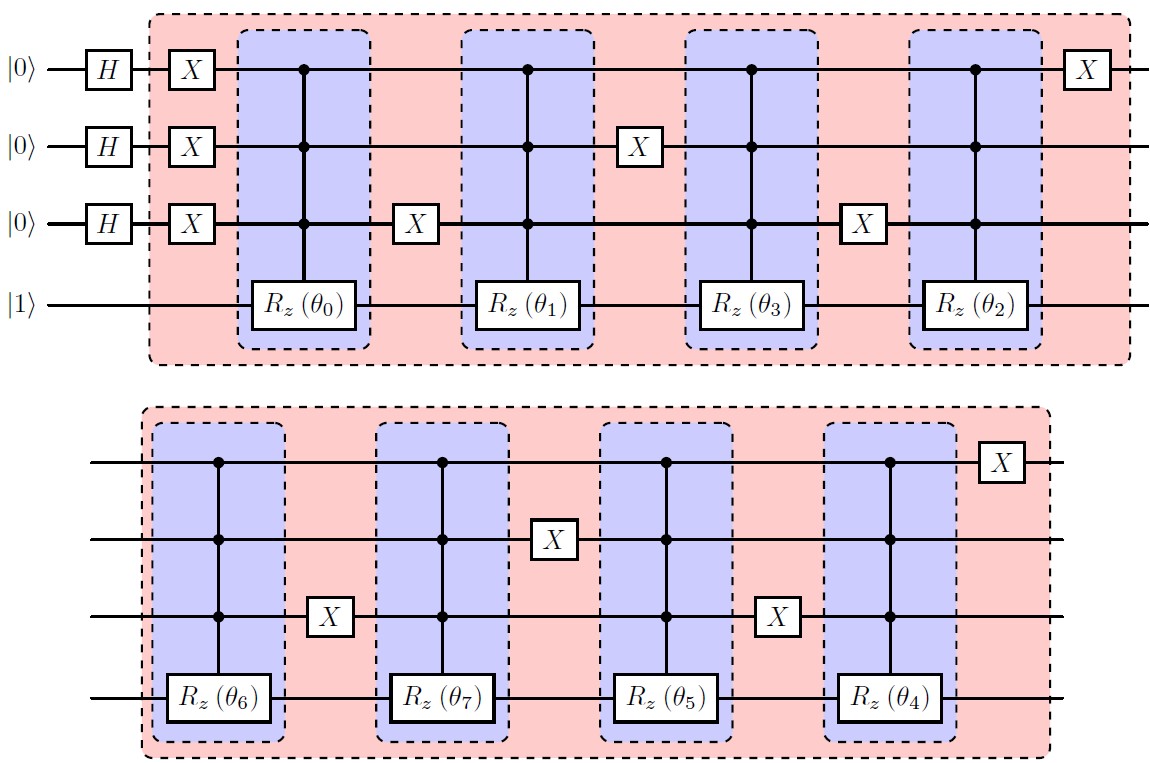}%[scale=0.49]
	\caption{Initial circuit for \(n = 4\).}
	\label{fig:circuit_for_n=4}
\end{figure}

\begin{figure}[h!]
	\centering
	\includegraphics[scale=0.38]{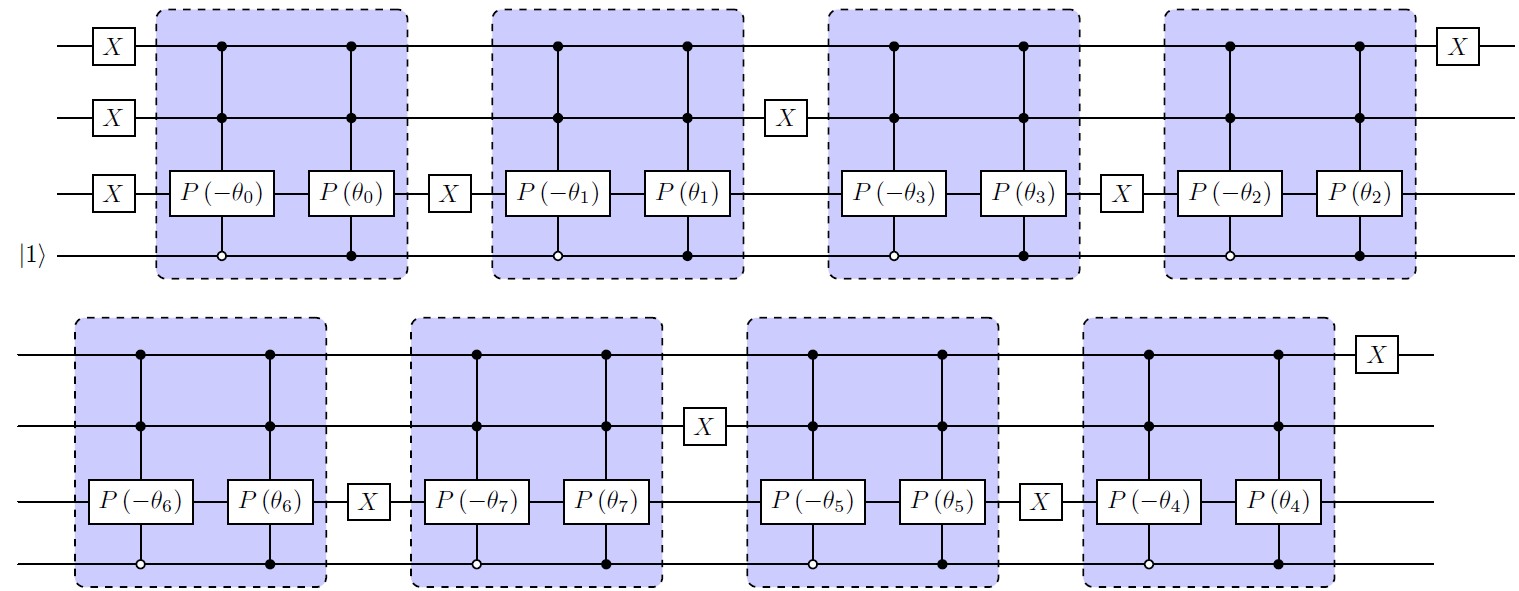}%[scale=0.49]
	\caption{The circuit after applying Lemma~\ref{lemma:lemma_1}.}
	\label{fig:after_application_of_lemma_1}
\end{figure}

\begin{figure}[h!]
	\centering
	\includegraphics[scale=0.25]{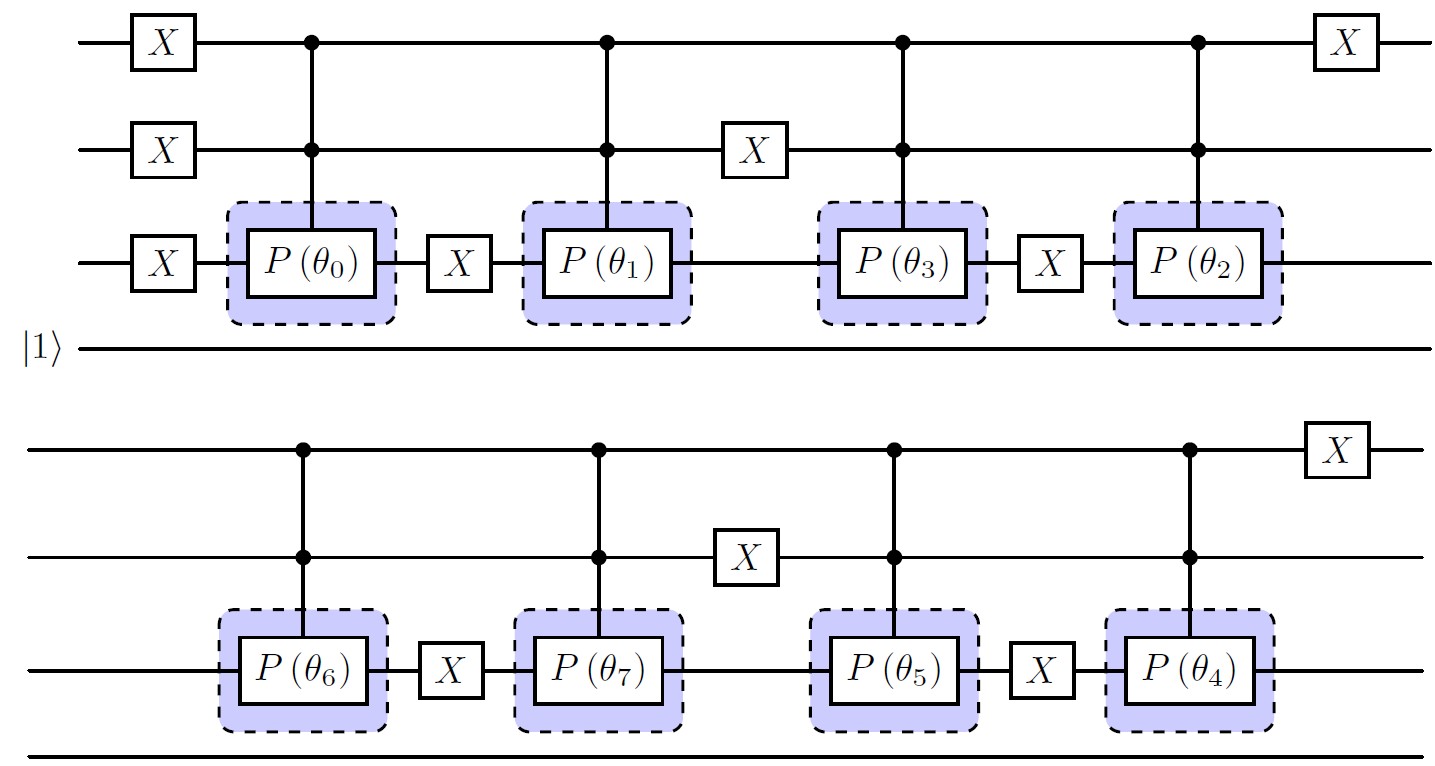}%[scale=0.33]
	\caption{The simplified circuit.}
	\label{fig:simplified_circuit}
\end{figure}

\begin{figure}[h!]
	\centering
	\includegraphics[scale=0.25]{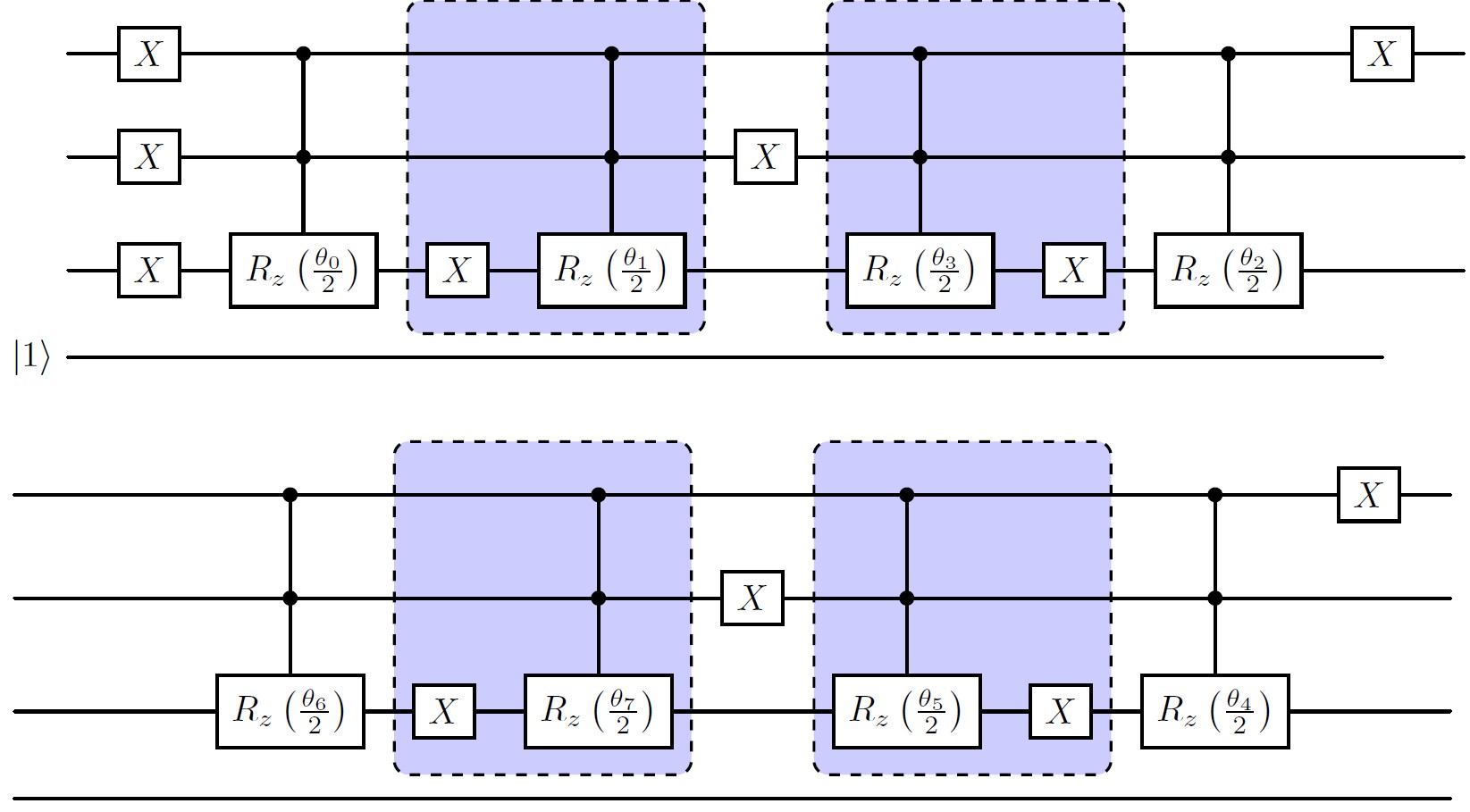}%[scale=0.33]
	\caption{The circuit after applying  Lemma~\ref{lemma:lemma_2}.}
	\label{fig:simplified_circuit_with_rotations}
\end{figure}

\begin{figure}[h!]
	\centering
	\includegraphics[scale=0.25]{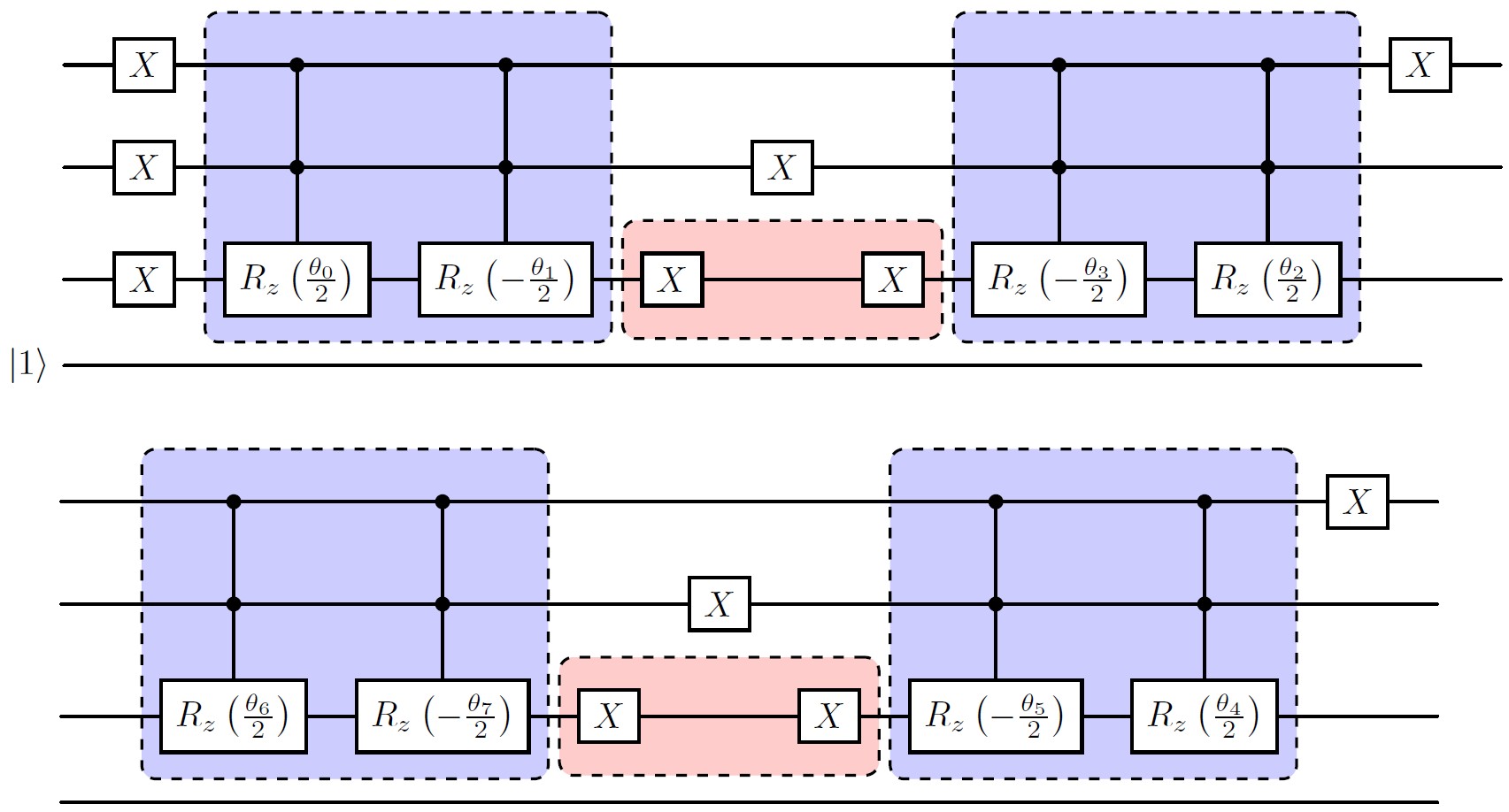}%[scale=0.33]
	\caption{The circuit after applying Lemma~\ref{lemma:lemma_3}.}
	\label{fig:merging_multi_qubit_controlled_rotations}
\end{figure}

\begin{figure}[h!]
	\centering
	\includegraphics[scale=0.35]{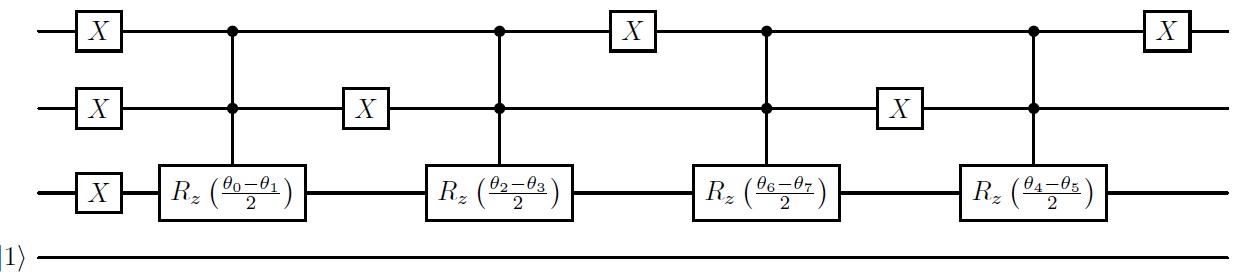}%[scale=0.47]
	\caption{The decomposition for a \(UCR^{n-2}_{z}\) gate using Gray code with additional \(X \) gate.}
	\label{fig:circuit_ucr_n-2_using_gray}
\end{figure}

\begin{figure}[h!]
	\centering
	\includegraphics[scale=0.4]{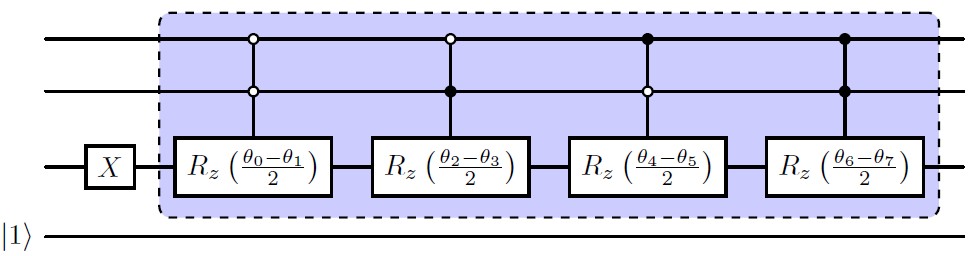}%[scale=0.54]
	\caption{The decomposition for the \(UCR^{n-2}_{z}\) gate using the standard binary code with additional \(X \) gate.}
	\label{fig:circuit_ucr_n-2}
\end{figure}

\begin{figure}[h!]
	\centering
	\includegraphics[scale=0.4]{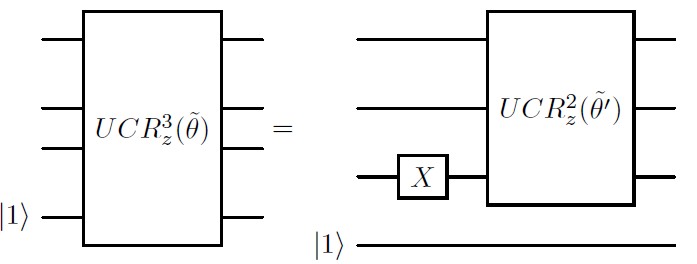}%[scale=0.54]
	\caption{Transformation of a \(UCR^{3}_{z}\) gate into a \(UCR^{2}_{z}\) gate.}
	\label{fig:ucr_identity_small}
\end{figure}

In general case, we get the circuit equivalence shown in Figure~\ref{fig:ucr_identity}. Using the technique from Section~\ref{sec:efficient_decomposition_of_ucr}, we obtain the decomposition consisting of \(2^{n-2}\) \(CNOT\) gates, whereas the original circuit on the left in Figure~\ref{fig:ucr_identity} requires \(2^{n-1}\) \(CNOT\) gates to decompose it. Finally, we get the circuit shown in Figure~\ref{fig:final_result} with modified angles \(\theta'_{i} = (\theta_{2i} - \theta_{2i+1})/2\). Thus, the target qubit in the state \(\ket{1}\) is eliminated.

\begin{figure}[h!]
	\centering
	\includegraphics[scale=0.3]{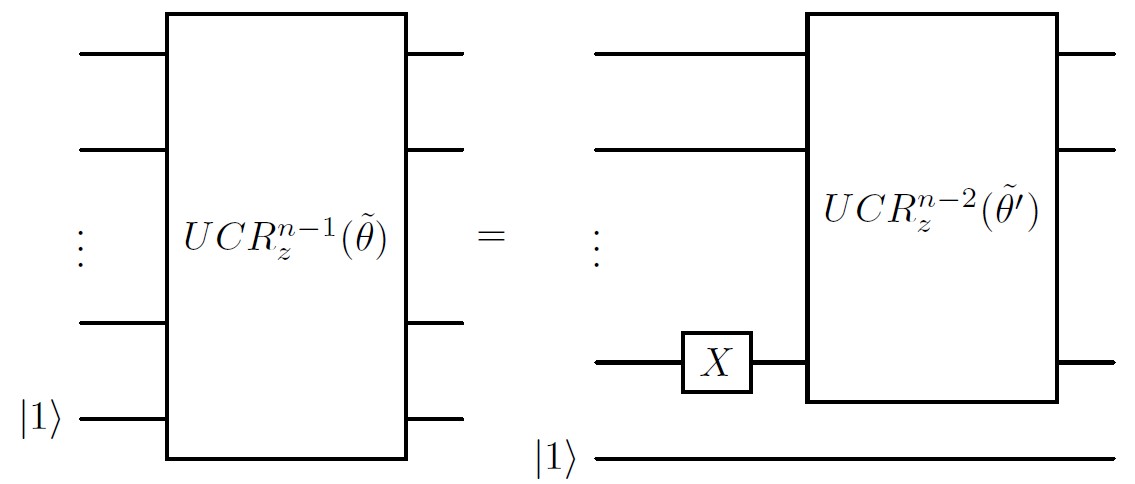}%[scale=0.38]
	\caption{Transformation of a \(UCR^{n-1}_{z}\) gate into a \(UCR^{n-2}_{z}\) gate.}
	\label{fig:ucr_identity}
\end{figure}
\end{proof}

\section{Efficient algorithm for quantum hashing: circuit depth vs angle precision}\label{sec:result2}

In this section, we propose an algorithm that offers a trade-off between the number of \(CNOT\) gates (and, consequently, the circuit depth) and the precision of angles. This is particularly important in the context of NISQ devices, as hardware-imposed angle precision limit remains a critical aspect.

Let us recall the generalization formula (Equation~(\ref{eq:qh_generalization})) for quantum hashing:
%\[
$
\ket{\psi(x)} = \frac{1}{\sqrt{d}}\sum_{j=0}^{d-1}\ket{j}\left(R_{a}\left(\theta_{j}\right)\ket{q_{n}}\right),
$
%\]
where \(\theta_{j} = \frac{4\pi s_{j}x}{q}\). Let us also recall that to implement quantum hashing, the circuit shown in Figure~\ref{fig:general_algorithm_for_qh} is used.

\begin{figure}[h!]
	\centering
	\includegraphics[scale=0.3]{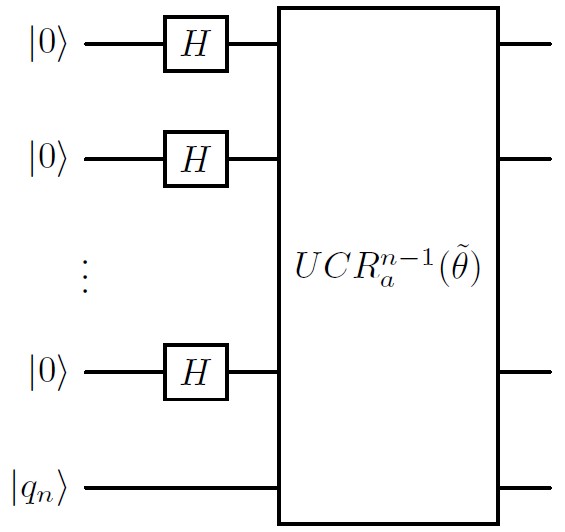}%[scale=0.43]
	\caption{General algorithm for quantum hashing.}
	\label{fig:general_algorithm_for_qh}
\end{figure}

Our goal is to efficiently decompose the \(UCR^{n-1}_{a}\) gate. To achieve this, we employ the technique described in Section~\ref{sec:efficient_decomposition_of_ucr}. In this case, the resulting circuit consists of an alternating sequence of \(2^{n-1}\) one-qubit rotations \(R_{a}\) and \(2^{n-1}\) \(CNOT\) gates. In the resulting circuit, the original angles \(\tilde{\theta}\) are modified, and the relationship between the original angles \(\tilde{\theta}\) and the modified angles \(\tilde{\theta'}\) is given by Equation~(\ref{eq:connection_between_angles}): \(\tilde{\theta'} = \frac{1}{d}M^{T}\tilde{\theta}\).

The precision of the original angles is \(O\left(\frac{1}{q}\right)\). Using Equation~(\ref{eq:connection_between_angles}), we derive that the precision of the modified angles becomes \(O\left(\frac{1}{dq}\right) = O\left(\frac{1}{q\log q}\right)\). Consequently, the modified angles are more sensitive. Implementing a circuit with such parameters can be challenging for current NISQ devices \cite{koczor2024}.

\subsection{Multi-qubit controlled rotation decomposition}

It is known \cite[Lemma~4.1]{barenco1995} that any special unitary matrix \(U\) (i.e.  unitary matrix with unity determinant, briefly, \(U \in SU(2)\)) can be represented in the following way:

\begin{equation}\label{eq:su_matrix_general_decomposition}
	U = R_{z}(\alpha)R_{y}(\beta)R_{z}(\gamma),
\end{equation}
where $\alpha, \beta, \gamma$ are some angles.

Any \(U \in SU(2)\) can be also expressed \cite[Lemma~4.3]{barenco1995} as 

\begin{equation}\label{eq:su_matrix_decomposition}
	U = AXBXC,
\end{equation}
where \(A, B, C\in SU(2)\) and \(ABC = I\). In Equation~(\ref{eq:su_matrix_decomposition}), using Equation~(\ref{eq:su_matrix_general_decomposition}), we set \(A = R_{z}(\alpha)R_{y}(\beta/2)\), \(B = R_{y}(-\beta/2)R_{z}(-(\alpha+\gamma)/2)\) and \(C = R_{z}((\gamma-\alpha)/2)\).

Let us consider the decomposition for an \(n\)-qubit controlled \(U\) gate (see Figure~\ref{fig:decomposition_for_multi_qubit_controlled_gate}) from \cite[Lemma~7.9]{barenco1995} that is based on Equation~(\ref{eq:su_matrix_decomposition}). In Equation~(\ref{eq:su_matrix_decomposition}), set \(A = R_{z}(\theta/2)\), \(B =R_{z}(-\theta/2)\) and \(C = I\) for \(R_{z}(\theta)\) and \(A = R_{y}(\theta/2)\), \(B =R_{y}(-\theta/2)\) and \(C = I\) for \(R_{y}(\theta)\). So, we obtain the decomposition shown in Figure~\ref{fig:decomposition_for_multi-qubit_controlled_rotation}.

\begin{figure}[h!]
	\centering
	\includegraphics[scale=0.37]{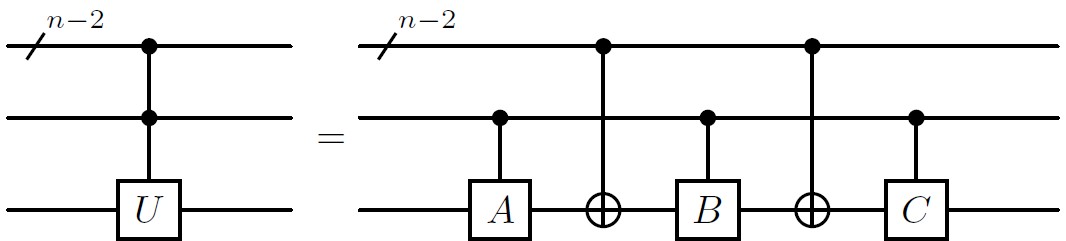}
	\caption{Decomposition for a \(C^{n-1}(U)\) gate.}
	\label{fig:decomposition_for_multi_qubit_controlled_gate}
\end{figure}

\begin{figure}[h!]
	\centering
	\includegraphics[scale=0.37]{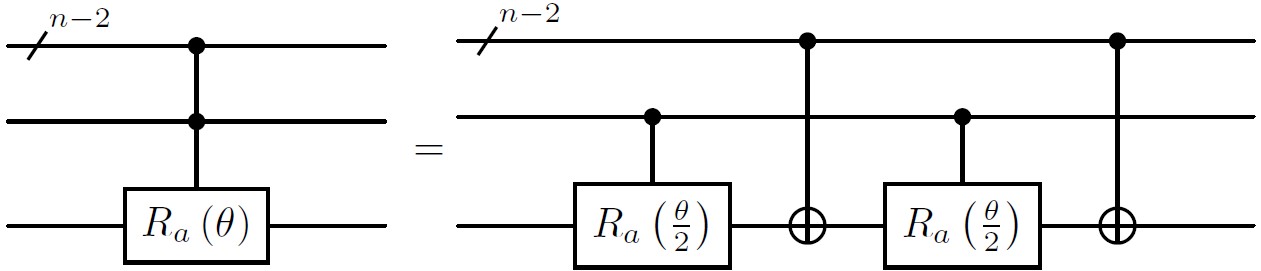}
	\caption{Decomposition for an \(n\)-qubit controlled rotation \(R_{a}\).}
	\label{fig:decomposition_for_multi-qubit_controlled_rotation}
\end{figure}

Note that the right circuit in Figure~{\ref{fig:decomposition_for_multi-qubit_controlled_rotation}} can be mirrored \cite{barenco1995}. It is worth mentioning that the qubit with index \(n-1\) can be used as an ancilla for decomposition of \(C^{n-2}(X)\) gates. It is known \cite[Corollary~7.4]{barenco1995} that \(C^{n}(X)\) gate can be decomposed into a circuit consisting of \(24n - 52\) \(CNOT\) gates.

A \(C^{1}(R_{a})\) gate is decomposed using the circuit presented in Figure~\ref{fig:decomposition_for_2_qubit_controlled_rotation}.

\begin{figure}[h!]
	\centering
	\includegraphics[scale=0.33]{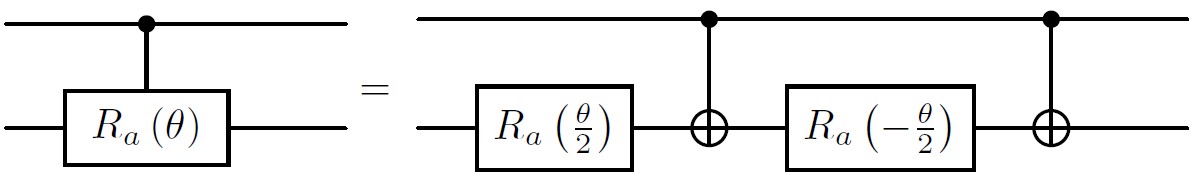}
	\caption{Decomposition for a \(R_{a}(\theta)\) gate.}
	\label{fig:decomposition_for_2_qubit_controlled_rotation}
\end{figure}

%%%%%%%%%%%%%

\subsection{Algorithm: circuit depth vs angle precision}

Let \textit{Decomposition~\(1\)} and \textit{Decomposition~\(2\)} be some procedures that will be described later. We propose the following algorithm:
%\begin{enumerate}
%	\item 
    (i) Recursively apply \textit{Decomposition~\(1\)} \(k\) times to construct a circuit consisting of \(2^{k}\) \(UCR^{n-k-1}_{a}\) gates and \(2^{k}\) \(CNOT\) gates.
	%\item 
    (ii) Apply \textit{Decomposition~\(2\)} to each of the \(2^{k}\) \(UCR^{n-k-1}_{a}\) gates.
%\end{enumerate}

\textit{Decomposition~\(1\)} involves applying the decomposition step shown in Figure~\ref{fig:ucr_decomposition_step} to each \(UCR\) gate.

Below, we describe \textit{Decomposition~\(2\)} procedure.
First, we reorder multi-qubit controlled rotations in a \(UCR^{n-k-1}_{a}\) gate using Gray code. Next, we highlight \(2^{n-k-3}\) circuit segments shown in Figure~\ref{fig:circuit_segment}.

\begin{figure}[h!]
	\centering
	\includegraphics[scale=0.3]{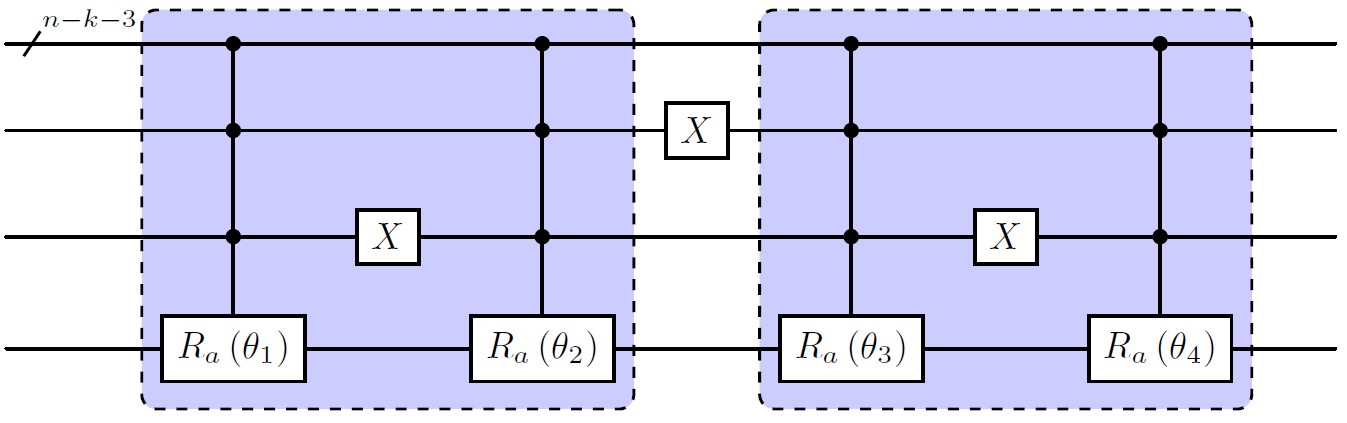}%[scale=0.38]
	\caption{The circuit segment.}
	\label{fig:circuit_segment}
\end{figure}

Let us propose a decomposition for the border-boxed circuit segments in Figure~\ref{fig:circuit_segment}. Thus, we focus on decomposing the smaller segment presented in Figure~\ref{fig:pair_of_multi-qubit_controlled_rotations}.

\begin{figure}[h!]
	\centering
	\includegraphics[scale=0.3]{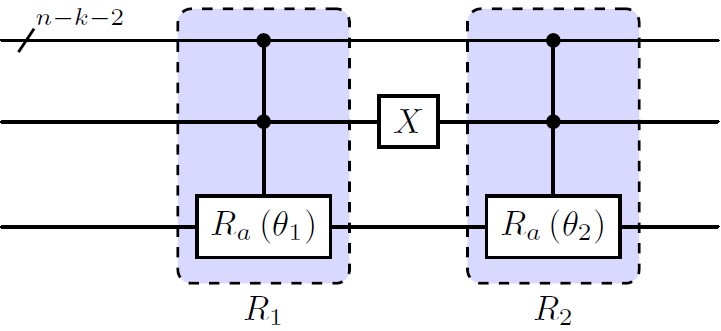}%[scale=0.43]
	\caption{A pair of multi-qubit controlled rotations.}
	\label{fig:pair_of_multi-qubit_controlled_rotations}
\end{figure}

% In \cite{barenco1995}, a decomposition of a multi-qubit controlled rotation is presented (for more details see Figure~\ref{fig:decomposition_for_multi-qubit_controlled_rotation} in Appendix~\ref{apx:appendix_decomposition}).
Recall the decomposition that is presented in Figure~\ref{fig:decomposition_for_multi-qubit_controlled_rotation}. We apply this decomposition to the \(R_{1}\) gate and its mirrored version to the \(R_{2}\) gate. As a result, it produces the circuit depicted in Figure~\ref{fig:decomposition_for_multi-qubit_controlled_rotations}. It can be seen that border-boxed \(C^{n-k-2}(X)\) gates in Figure~\ref{fig:decomposition_for_multi-qubit_controlled_rotations} cancel out each other, and we obtain the circuit shown in Figure~\ref{fig:simplified_decomposition_of_multi-qubit_controlled_rotations}. We decompose the blue border-boxed fragment in Figure~\ref{fig:simplified_decomposition_of_multi-qubit_controlled_rotations} using the circuit shown in Figure~\ref{fig:decomposition_for_the_segment}. Next, we apply the decomposition shown in Figure~\ref{fig:simplified_decomposition_of_multi-qubit_controlled_rotations} to border-boxed circuit segments depicted in Figure~\ref{fig:circuit_segment}.  See Figure~\ref{fig:bigger_circuit_segment} and note that border-boxed controlled rotations are merged into one two-qubit controlled rotation. This trick is done for all pairs of `outside' two-qubit controlled rotations.

\begin{figure}[h!]
	\centering
	\includegraphics[scale=0.28]{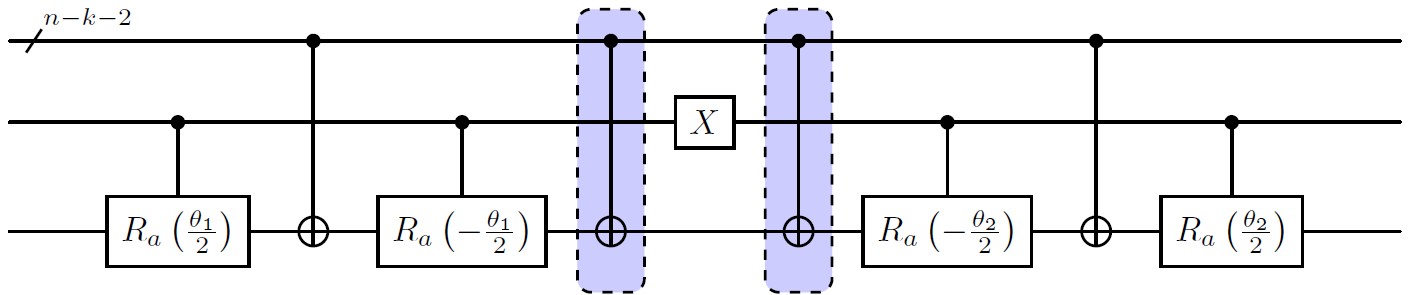}%[scale=0.37]
	\caption{Decomposition for a pair of multi-qubit controlled rotations.}
	\label{fig:decomposition_for_multi-qubit_controlled_rotations}
\end{figure}

\begin{figure}[h!]
	\centering
	\includegraphics[scale=0.3]{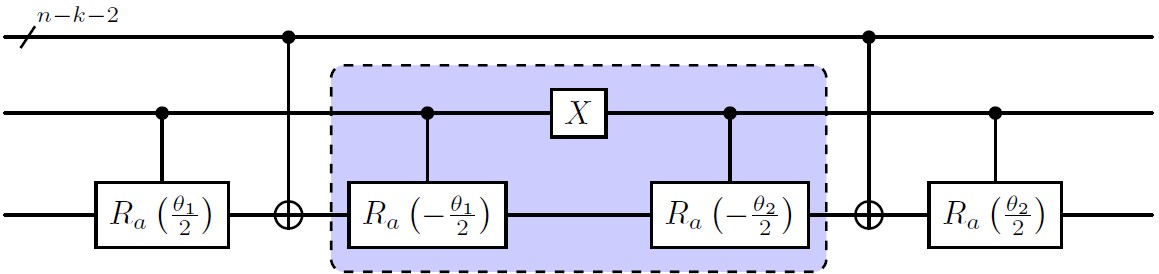}%[scale=0.42]
	\caption{Simplified decomposition for a pair of multi-qubit controlled rotations.}
	\label{fig:simplified_decomposition_of_multi-qubit_controlled_rotations}
\end{figure}

\begin{figure}[h!]
	\centering
	\includegraphics[scale=0.3]{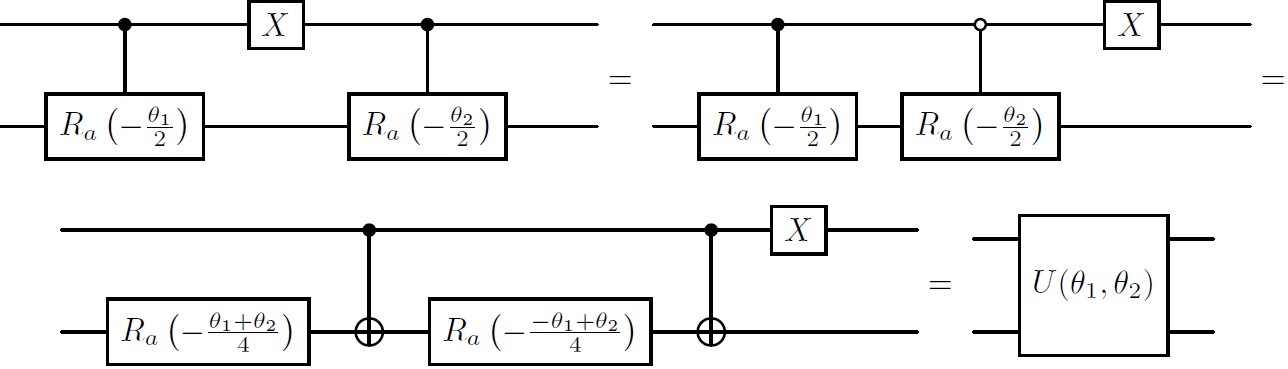}%[scale=0.43]
	\caption{Decomposition for the circuit fragment.}
	\label{fig:decomposition_for_the_segment}
\end{figure}

\begin{figure}[h!]
	\centering
	\includegraphics[scale=0.3]{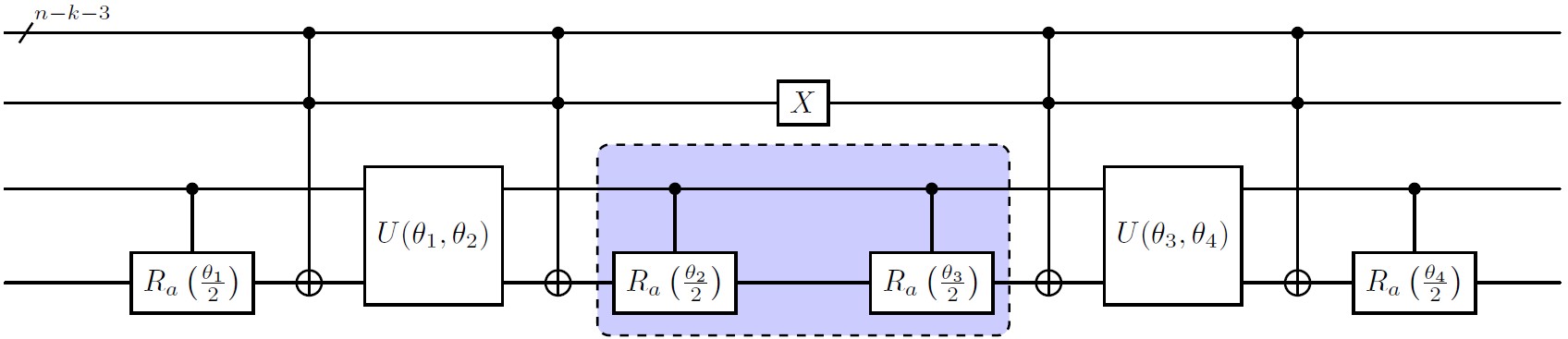}%[scale=0.43]
	\caption{Intermediate decomposition for the circuit segment.}
	\label{fig:bigger_circuit_segment}
\end{figure}

Now, let us count the total number of \(CNOT\) gates in the resulting circuit. The decomposition of \(C^{n-k-2}(X)\) \cite{barenco1995} requires
%\[
$
24\cdot(n-k-2)-52 = 24\cdot(n-k)-100
$
%\]
\(CNOT\) gates. Note that each of the \(2^{n-k-2}\) circuit segments includes two \(C^{n-k-2}(X)\) gates, one \(U_{a}(\theta_{1}, \theta_{2})\) gate, and one \(C^{1}(R_{a})\) gate. In addition, the first \(C^{1}(R_{a})\) gate must be taken into account. Thus, the decomposition of each \(UCR^{n-k-1}_{a}\) gate requires
$2^{n-k-2} \cdot (2 \cdot (24 \cdot (n-k)-100)+2+2) + 2=2^{n-k} \cdot (12 \cdot (n-k)-49)+2$
%\[
%\begin{split}
%	&2^{n-k-2} \cdot (2 \cdot (24 \cdot (n-k)-100)+2+2) + 2=\\
%	&2^{n-k} \cdot (12 \cdot (n-k)-49)+2
%\end{split}
%\]
\(CNOT\) gates.

So, the total number of \(CNOT\) gates after \(k\) iterations is the following:
$
3 \cdot 2^{k} + 2^{n} \cdot (12 \cdot (n-k)-49)
$
or
$
3 \cdot 2^{k} + d \cdot \left(24 \cdot (\log d-k)-74\right),
$
where \(k \leq n - 5 = \log d - 4\).

\begin{theorem}
	Application of the proposed decomposition algorithm to the original circuit that implements quantum hashing yields a circuit with \(3 \cdot 2^{k} + 2^{n} \cdot (12 \cdot (n-k)-49)\) or \(3 \cdot 2^{k} + d \cdot \left(24 \cdot (\log d-k)-74\right)\) \(CNOT\) gates, for \(k \leq n - 5 = \log d - 4\), where \(k\) is the number of applications of \textit{Decomposition~\(1\)}.
\end{theorem}

We can see that as \(k\) increases, the circuit depth decreases from \(O(\log q \log\log q)\) to \(O(\log q)\), while the angle precision changes from \(O\left(1/q\right)\) to \(O\left(1/(q\log q)\right)\).

%We conclude that as \(k\) increases, the circuit depth decreases from \(O(\log q \log\log q)\) to \(O(\log q)\), while the angle precision increases from \(O\left(1/q\right)\) to \(O\left(1/(q\log q)\right)\).

Authors of \cite{khadieva2024} propose an algorithm  that balances between the number of \(CNOT\) gates and the precision of rotation angles that yields a circuit with \(2^{k}+d(96(\log d-k)-384))\) \(CNOT\) gates for \(k \leq \log d - 5\), matching the asymptotic complexity of our algorithm. Thus, our approach reduces the number of \(CNOT\) gates by
$
d(72(\log d-k)-310) - 2^{k+1} \geq \frac{799}{16}d
$ for \(k \leq \log d - 5\),
achieving asymptotic savings of \(O(d)\) \(CNOT\) gates.

\subsection{Circuit optimization for the phase form of quantum hashing}

Note that the circuit depicted in Figure~\ref{fig:simplified_decomposition_of_multi-qubit_controlled_rotations} can be further simplified if we focus on the phase form of quantum hashing. A controlled rotation about the \(z\)-axis can be replaced by a single rotation about the \(z\)-axis applied to the control qubit provided that the target qubit is fixed in one of the basis states. This is illustrated in Figure~\ref{fig:circuit_equivalences} for the case, when the target qubit is in the state \(\ket{1}\).

\begin{figure}[h!]
	\centering
	\includegraphics[scale=0.65]{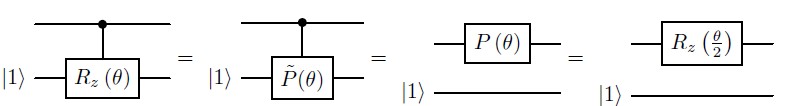}
	\caption{Circuit equivalences.}
	\label{fig:circuit_equivalences}
\end{figure}

Thus, the `outside' controlled rotations in Figure~\ref{fig:simplified_decomposition_of_multi-qubit_controlled_rotations} are replaced with single rotations, as described above, leading us to the circuit shown in Figure~\ref{fig:improvement_phase_form}.

\begin{figure}[h!]
	\centering
	\includegraphics[scale=0.65]{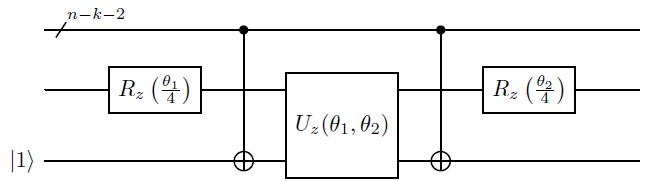}
	\caption{Simplified decomposition for a pair of multi-qubit controlled rotations.}
	\label{fig:improvement_phase_form}
\end{figure}

Then, the circuit depicted in Figure~\ref{fig:bigger_circuit_segment} is further simplified, as shown in Figure~\ref{fig:bigger_circuit_segment_for_phase}.

\begin{figure}[h!]
	\centering
	\includegraphics[scale=0.55]{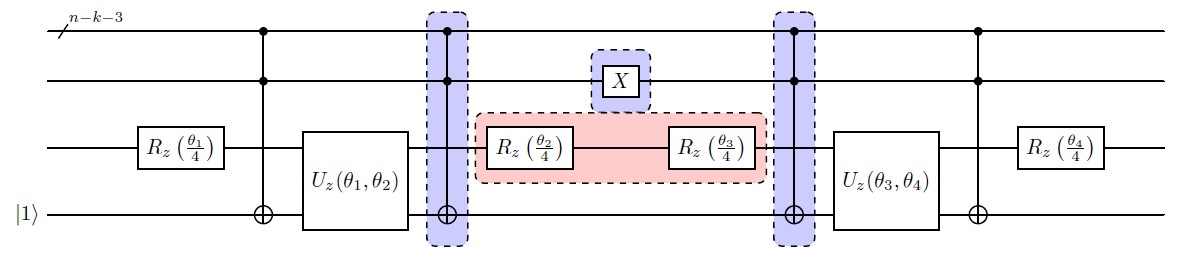}
	\caption{Intermediate decomposition for the segment of the circuit implementing the phase form of quantum hashing.}
	\label{fig:bigger_circuit_segment_for_phase}
\end{figure}

Using the circuit identity shown in Figure~\ref{fig:pair_of_multi-qubit_controlled_x_gates} for blue border-boxed gates we obtain the circuit depicted in Figure~\ref{fig:bigger_circuit_segment_for_phase2}.

\begin{figure}[h!]
	\centering
	\includegraphics[scale=0.6]{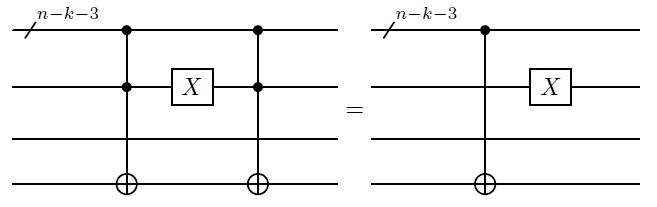}
	\caption{A pair of multi-qubit controlled \(X\) gates.}
	\label{fig:pair_of_multi-qubit_controlled_x_gates}
\end{figure}

\begin{figure}[h!]
	\centering
	\includegraphics[scale=0.6]{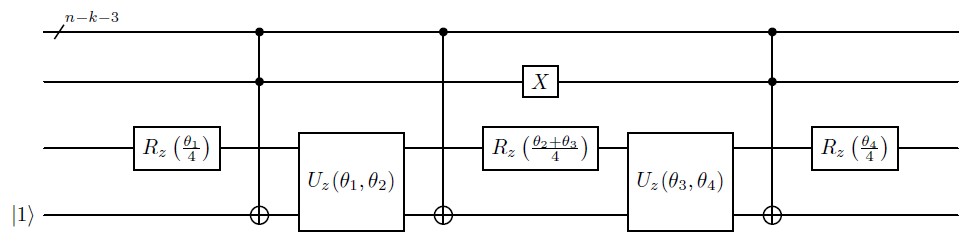}
	\caption{Intermediate decomposition for the segment of the circuit implementing the phase form of quantum hashing.}
	\label{fig:bigger_circuit_segment_for_phase2}
\end{figure}

Each of the \(2^{n-k-2}\) circuit segments includes one \(C^{n-k-3}(X)\) gate, except the last segment, and one \(U_{z}(\theta_{1}, \theta_{2})\) gate. It is necessary to mention that the last segment consists of a \(C^{n-k-2}(X)\) gate. We also need to account for the first \(C^{n-k-2}(X)\) gate. Thus, the decomposition of each \(UCR^{n-k-1}_{z}\) gate requires

\[
\begin{split}
	&2 \cdot (24 \cdot (n-k)-100) + (2^{n-k-2}-1) \cdot (24 \cdot (n-k-3)-52) + 2^{n-k-2} \cdot 2=\\
	&2^{n-k-1} \cdot (12 \cdot (n-k)-61) + 24 \cdot (n-k) - 76
\end{split}
\]
\(CNOT\) gates.

Finally, the total number of \(CNOT\) gates after \(k\) iterations is the following:
\[
2^{k} \cdot (24 \cdot (n-k) - 75) + 2^{n-1} \cdot (12 \cdot (n-k)-61)
\]
or
\[
2^{k} \cdot (24 \cdot (\log d-k) - 51) + d \cdot (12 \cdot (\log d-k)-49),
\]
where \(k \leq n - 6 = \log d - 5\).

\section{Conclusion}\label{sec:conclusions}

Implementation of algorithms on current quantum computers faces significant challenges. Therefore, it is essential to design efficient algorithms optimized with respect to multiple metrics including the circuit depth, the number of two-qubit gates, the precision of rotation angles.

We propose algorithms for quantum hashing that are inspired by NISQ quantum devices. First, we demonstrated how to eliminate an ancilla qubit used in constructing the phase form quantum hash. Our technique reduces the number of qubits used by one and reduces the number of \(CNOT\) gates and the circuit depth by half, thereby saving both time and space resources. Our technique also outperforms all existing methods. Second, we present an algorithm that provides a trade-off between the number of \(CNOT\) gates (the circuit depth) and the rotation angle precision. Since there are limitations on the precision of rotation angles, our algorithm allows the circuit to avoid impractically small rotations. Although asymptotically our algorithm is equivalent to the one from \cite{khadieva2024}, it is more efficient in terms of the number of \(CNOT\) gates.

%\section{Acknowledgements}
%The research has been supported by Russian Science Foundation Grant 24-21-00406, \url{https://rscf.ru/project/24-21-00406/}.

%\begin{credits}
\textbf{\ackname} The research has been supported by Russian Science Foundation Grant 24-21-00406, \url{https://rscf.ru/project/24-21-00406/}
%\subsubsection{\discintname}
%It is now necessary to declare any competing interests or to specifically state that the authors have no competing interests. Please place the statement with a bold run-in heading in small font size beneath the (optional) acknowledgments\footnote{If EquinOCS, our proceedings submission system, is used, then the disclaimer can be provided directly in the system.}, for example: The authors have no competing interests to declare that are
%relevant to the content of this article. Or: Author A has received research
%grants from Company W. Author B has received a speaker honorarium from
%Company X and owns stock in Company Y. Author C is a member of committee Z.
%\end{credits}
\bibliographystyle{splncs04}
\bibliography{generic}
\newpage
\appendix

\end{document}